\documentclass[a4paper, 11pt]{article}
\usepackage{latexsym,amsmath,amsfonts,amssymb,amsbsy,amsopn,amstext}
\usepackage{epsfig,epstopdf,graphicx}
\usepackage[mathscr]{eucal}
\usepackage{color,multicol,makeidx}
\usepackage{bm,mathrsfs}
\usepackage{float}
\usepackage{cite}
\usepackage{natbib}
\usepackage{mathdots}
\usepackage{enumerate}
\usepackage{hyperref}

\def\dref#1{\eqref{#1}}
\def\eq{\displaystyle\stackrel\triangle=}
\def\xra{\xrightarrow}

\newcommand{\diag}{\ensuremath{\mathrm{diag}}}
\DeclareMathOperator*{\rank}{rank} \DeclareMathOperator*{\tr}{Tr}
\DeclareMathOperator*{\argmin}{arg\,min}

\newcommand*\mcapinn[2]{\vcenter{\hbox{$\mathsurround=0pt
  \ifx\displaystyle#1\textstyle\else#1\fi\bigcap$}}}

\newcommand*\mcupinn[2]{\vcenter{\hbox{$\mathsurround=0pt
  \ifx\displaystyle#1\textstyle\else#1\fi\bigcup$}}}

\DeclareFontFamily{OT1}{pzc}{}
\DeclareFontShape{OT1}{pzc}{m}{it}{<-> s * [1.200] pzcmi7t}{}
\DeclareMathAlphabet{\mathpzc}{OT1}{pzc}{m}{it}

\graphicspath{{./images/}}
\newtheorem{exam}{Example}

\newcommand{\bra}[1]{\ensuremath{|#1\rangle}}

\newcommand{\braket}[2]{\ensuremath{|#1\rangle\langle#2|}}

\newenvironment{proof}[1][Proof]{\noindent\textbf{#1.} }{\hfill \rule{0.5em}{0.5em}}

\newtheorem{thm}{Theorem}

\newtheorem{lem}{Lemma}
\newtheorem{rem}{Remark}
\newtheorem{prop}{Proposition}

\begin{document}
	
\title{\bf Quantum Tomography by Regularized Linear Regression}

\author{Biqiang Mu\thanks{Key Laboratory of Systems and Control, Institute of Systems Science, Academy of Mathematics and Systems Science, Chinese Academy of Sciences, Beijing 100190, China. E-mail: bqmu@amss.ac.cn.}, Hongsheng Qi\thanks{Key Laboratory of Systems and Control, Institute of Systems Science, Academy of Mathematics and Systems Science, Chinese Academy of Sciences, Beijing 100190, China; Research School of Electrical, Energy and Materials Engineering,  The Australian National University, Canberra 0200, Australia. E-mail: qihongsh@amss.ac.cn.}, Ian R. Petersen\thanks{Research School of Electrical, Energy and Materials Engineering,  The Australian National University, Canberra 0200, Australia. E-mail: ian.petersen@anu.edu.au.}, Guodong Shi\thanks{Australian Center for Field Robotics, School of Aerospace, Mechanical and Mechatronic Engineering, The University of Sydney, NSW 2008, Australia. E-mail: guodong.shi@sydney.edu.au.}}

\date{}
	
\maketitle

\begin{abstract}
In this paper, we study extended    linear regression approaches for quantum state tomography based on regularization techniques.  For unknown quantum states  represented by density matrices, performing measurements under certain basis yields random outcomes, from which a classical linear regression model can be established. First of all,  for complete or over-complete measurement bases,  we show that the empirical data can be utilized for the construction of a weighted least squares estimate (LSE) for quantum tomography. Taking into consideration the trace-one condition, a constrained 
weighted LSE can be explicitly computed, being the optimal unbiased estimation among all linear estimators. Next, for general measurement bases, we show that  $\ell_2$-regularization   with proper regularization gain  provides even lower mean-square error under a cost in bias. The regularization parameter is tuned by two estimators in terms of a risk characterization. Finally, a concise and unified  formula is established for the regularization parameter with complete measurement basis under an equivalent regression model, which proves that the proposed tuning estimators are asymptotically optimal as the number of samples grows to infinity under the risk metric.  Additionally, 
numerical examples are provided to validate the established  results.
\end{abstract}
	

\section{Introduction}
Since Feynman pointed out the possibility of using quantum resources to carry out computation in the early 1980s, significant progresses have been made in both the theoretical understanding and the  real-world implementations for computing and communication mechanisms based on quantum states \citep{Nielsen}. Underpinning such efforts lies in the development of quantum tomography \citep{1989,qubit2001,art2005,science2014},  where reliable quantum state estimation \citep{prl2009,prl2010,prl2011} and system identification methods \citep{rouchon2009,rouchon2012,ian2018,xue2018,ian2019} provide basic  assurance for the validity of quantum systems that we intend to work on. The fundamental  quantum measurement postulate indicates that any form of  quantum information probe    would have an inherent probabilistic nature. The exponentially growing complexity of quantum systems along with the increasing scale further adds to the challenging reality: only partial information can be made available via measurements for uncertain quantum systems; processing the measurement data faces enormously high computation barrier for large-scale quantum systems.

One primary task of quantum tomography is to determine an unknown quantum state from a number of identical copies \citep{prl2009,prl2010,prl2011}. Performing measurement on those copies   along certain observables, i.e., measurement bases,  yields independent realizations of some hidden   random variable whose statistics encode  the quantum state and the observable. Therefore, utilizing the outcomes of the measurements  we can build estimations of the unknown quantum state since the observables are known (which can be selected and designed). 
Apparently the choice of the estimation method is not unique since the estimation  error metrics can be characterized by different metrics, and  the resulting computational feasibility also raises constraint on the potentially viable  estimation approaches. Therefore, there is often a tradeoff between estimation quality and computational efficiency \citep{2012}. 

Linear regression, as a universal estimator \citep{Ljung1999}, becomes a natural and important   quantum state tomography approach due to its  simplicity and practicability. A thorough comparison was made in  \cite{qi2013} between linear regression method  and maximum likelihood estimation for quantum state tomography under complete or over-complete measurement bases. Linear regression was also applied to quantum tomography with incomplete measurement bases for low-rank, e.g., \cite{gross2010,tit2011,alquier2013} or sparse, e.g., \cite{annals2016} quantum states  in view of the insights from compressed sensing, where a small number of measurement bases was proven to be enough  for the recovery of a high-dimensional quantum state with high probability as the dimension increases. Recently, linear regression method was also generalized to the adaptive measurement case where selection of the measurement basis depends on the previous measurement outcomes \citep{qi2017}. 

In this paper, we study  the role of regularization for linear regression-based quantum state tomography. In recent years, the power of regularization has been  well noted in the literature of optimization, machine learning, and system identification, e.g., \cite{online,dp2016,chiu2016} for the purpose of avoiding overfitting in empirical learning. Noting any quantum state can be  represented as a trace-one positive  Hermitian  density matrix, which is of low rank if it is a combination of a small number of pure states, we establish the following results. 
\begin{itemize}
\item For complete or over-complete measurement basis,  the empirical data can be utilized for constructing of a weighted least squares estimate (LSE) for quantum tomography. The weighted LSE provides reduced mean-square error compared to standard LSE. Taking into consideration the trace-one condition, the constrained 
weighted LSE can be explicitly computed, which is the optimal unbiased estimation that is linear in the measurement data. 
\item For any (complete, over-complete, or under-complete) measurement basis, a closed form solution is established for tomography with $\ell_2$-regularized  weighted linear regression.   It is shown that with proper regularization parameter, this regularized regression always provides even lower mean-square error subject to, of course, a price of additional bias. 
\item The regularization parameter can be further optimized subject to a risk characterization. An explicit formula is established for the regularization parameter under an equivalent regression model, which proves that the proposed tuning estimators are asymptotically  optimal for complete bases as the number of samples grows to infinity for the risk metric.  
\end{itemize}
Numerical examples are provided for the validation of  the established theoretical results, which confirms the potential usefulness of the proposed linear regression methods in quantum state tomography. 

The remainder of the paper is organized as follows. In Section \ref{secproblem}, we present the standard linear regression model for quantum state tomography, and review some preliminary knowledge on the underlying rationale. In Section \ref{secregular}, we present the  extended quantum  tomography  methods based on weighted LSE, constrained weighted LSE, and constrained regularized  LSE, respectively, whose performances in terms of mean-square error are thoroughly  investigated.   Section \ref{sec4} further presents the asymptotically optimal regularization gain under an equivalent model. Numerical examples are presented in Section \ref{sec5}, and finally, some concluding remarks are given in Section \ref{sec6}.

\section{Problem Definition and Preliminaries}\label{secproblem}
\subsection{Linear Regression for Quantum State Tomography}
Let $\mathcal{H}$ be a $d$-dimensional Hilbert space that characterizes the state space of a quantum system. Denote the space of  linear Hermitian  operators over $\mathcal{H}$ by $\mathcal{L}(\mathcal{H})$.
Suppose that $\big\{\mathsf{B}_i\big\}_{i=1}^{d^2}$ is an  orthonormal basis of $\mathcal{L}(\mathcal{H})$ with $\tr(\mathsf{B}_i^\dagger\mathsf{B}_j)=\delta_{ij}$ and $\mathsf{B}_i^\dagger = \mathsf{B}_i$, where 
$\tr(\cdot)$ means the trace of a square matrix, 
$(\cdot)^\dagger$ represents the Hermitian conjugate of a complex matrix, and  $\delta_{ij}$ is the Kronecker function. A quantum state $\rho$
as a density operator over $\mathcal{H}$ can then be expressed by
\begin{align}
\rho = \sum_{i=1}^{d^2}  \theta_i \mathsf{B}_i \label{rho}
\end{align}
where $\theta_i=\tr(\rho \mathsf{B}_i) \in\mathbb{R}$ is the coordinate of $\rho$ under the given basis $\big\{\mathsf{B}_i\big\}_{i=1}^{d^2}$. 
Let there be a positive operator-valued measurement (POVM) over the space $\mathcal{H}$, denoted by $\{\mathsf{M}_m \}_{m=1}^M$ 
with $\sum_{m=1}^M \mathsf{M}_m^\dagger \mathsf{M}_m =\mathsf{I}$, where  $\mathsf{I}$ is the identity operator. 
Then  $\mathsf{E}_m\eq \mathsf{M}_m^\dagger \mathsf{M}_m$  can be expressed as a linear combination of the orthogonal basis of $\big\{\mathsf{B}_i\big\}_{i=1}^{d^2}$:
$$
\mathsf{E}_m = \sum_{i=1}^{d^2} \beta_{mi} \mathsf{B}_i
$$ for each $1\leq m \leq M$, 
where $\beta_{mi}=\tr(\mathsf{E}_m \mathsf{B}_i)$. 
When the quantum state $\rho$ is being measured under the POVM $\{\mathsf{M}_m \}_{m=1}^M$, the probability of observing outcome $m$ is  $
p_m = \tr(\mathsf{E}_m \rho )
=\bm{\beta}_m^\top\bm{\theta} $, where $\bm{\beta}_m = [\beta_{m1},\cdots,\beta_{md^2}]^\top $ and $\bm{\theta} = [\theta_1,\cdots,\theta_{d^2}]^\top $. Denoting $\mathbf{p}=[p_1,\dots,p_M]^\top\in \mathbb{R}^M$, $\mathbf{A}=[\bm{\beta}_1,\dots,\bm{\beta}_M]^\top \in\mathbb{R}^{M\times d^{2}}$, we have the following  fundamental quantum measurement description in the form of a linear algebraic equation:
$$
\mathbf{p}= \mathbf{A} \bm{\theta}. 
$$ 
The tomography of an unknown quantum state $\rho$ is therefore equivalent to identifying the vector $\bm\theta$,  where $\mathbf{A}$ is known and $\mathbf{p}$ is estimated by experimental realizations of measuring $\rho$ from the POVM $\{\mathsf{M}_m \}_{m=1}^M$. The POVM can be in general represented under Pauli matrices, see e.g, \cite{annals2016,annals2013}.

A standard quantum state tomography process is as follows: (i) Prepare $N=nM$ identical copies of an uncertain quantum state $\rho$; (ii) Perform measurement  along each $\mathsf{M}_m$ within   the POVM $\{\mathsf{M}_m \}_{m=1}^M$ independently for $n$ copies; (iii)  For each $1\leq m\leq M$, record the number of times that the outcome $\mathsf{M}_m$ is observed among those $n$ experiments, denoted by $\#m$, from the $n$ experiments. Then,
\begin{align}
\widehat{p}_m=\frac{\#m}{n}\label{pm_est}
\end{align}
 is a natural estimator of the probability $p_m$, leading to
\begin{align}
\widehat{p}_m = {\bm{\beta}}_m^\top \theta + e_m,
\end{align}
where $e_m=\widehat{p}_m - p_m$ is the estimation error. The distribution of $e_m$ depends on the sample size $n$, as it is the sum of $n$ identical and independently distributed (i.i.d.) Bernoulli random variables with mean $p_m$. This naturally  yields the following   linear regression problem:
\begin{align} 
\mathbf{y}=\mathbf{A}\bm{\theta} + \mathbf{e} \label{rg}
\end{align}
with
$\mathbf{y}=[\widehat{p}_1,\cdots,\widehat{p}_M]^\top $
and $\mathbf{e}=(e_1,\cdots,e_M)^\top $.

\subsection{The Noise Distribution} 
Define i.i.d. random variables $b_{l}^{(m)}$ for $1\leq l \leq n$, which takes value $1$ with  probability $p_m$ and $0$ with probability $1-p_m$.
Then there holds
\begin{align}
e_m=\widehat{p}_m -p_m= \frac{\sum_{l=1}^{n}  b_{l}^{(m)}}{n}-p_m
=\sum_{l=1}^n \frac{ b_{l}^{(m)}-p_m }{n}. 
\end{align}
Note that $(b_{l}^{(m)}-p_m)/n$ takes value $(1-p_m)/n$ with  probability $p_m$ and $-p_m/n$ with probability $1-p_m$.
It follows that
	\begingroup
\allowdisplaybreaks
\begin{subequations}
	\label{em_est}
	\begin{align}
&\mathbb{E}(e_m )=0\\
&\mathbb{V}(  e_m ) 
=\mathbb{E} (e_m )^2
=   (p_m-p_m^2)/n. \label{varest}
\end{align}
\end{subequations}
\endgroup
As a result, the distribution of $e_m$ is as follows:
\begin{align}
\mathbb{P}\left(  e_m = \left(\frac{1-p_m}{n} \right)^k \left( -\frac{p_m}{n}\right)^{n-k}  \right)
=\dbinom{n}{k} p_m^k (1-p_m)^{n-k}.
\end{align}
 As $n$ tends to infinity, each $e_m$ will converge to a Gaussian random variable with mean $0$ and variance $(p_m-p_m^2)/n$.

 \subsection{Simultaneous Measurements}
 In the tomography process described above, each $\mathsf{M}_m$ is separately measured, i.e., a binary outcome is recorded for any copy of $\rho$, where $1$ represents  $\mathsf{M}_m$, and $0$ represents  $\mathsf{I}-\mathsf{M}_m$. 
An alternative quantum tomography process can be described  based on $n$ copies of $\rho$, where we perform measurement by the POVM $\{\mathsf{M}_m \}_{m=1}^M$  collectively. To be precise, the outcome associated with each copy of the quantum state  now takes value in $\{1,\dots,m\}$, and then the number of times that the outcome $m$ is observed among those $n$ experiments, denoted by $\#'m$, is recorded  from the $n$ experiments for each $1\leq m\leq M$. Consequently,
\begin{align}
\bar{p}_m=\frac{\#'m}{n}\label{pm_est2}
\end{align}
 is still an estimator of the probability $p_m$, leading to
\begin{align}
\bar{p}_m = {\bm{\beta}}_m^\top \theta + \bar{e}_m,\ \ m=1,\dots,M. 
\end{align}
The estimation error $\bar{e}_m$ as a random variable has the same distribution of $e_m$. However, the $\bar{e}_m$ are no longer independent since now $\sum_{m=1}^n \bar{p}_m=1$ is a sure event. Except for this minor difference, this new formulation of quantum tomography procedure remains the same.

\subsection{Standard Least Squares}
For the estimation  problem \eqref{rg}, the least squares (LS) solution 
\begin{align}
\nonumber
\widehat{\bm{\theta}}^{\rm LS}
& =\argmin_{\bm{\theta}} \  (\mathbf{y}-\mathbf{A} \bm{\theta})^\top  (\mathbf{y}-\mathbf{A} \bm{\theta})\\
&=(\mathbf{A}^\top  \mathbf{A})^{-1} \mathbf{A}^\top \mathbf{y} \label{ls}
\end{align}
is a common choice provided that $\mathbf{A}$ has full column rank. 
The estimate $\widehat{\bm{\theta}}^{\rm LS}$ admits   the following   properties:
\begin{itemize}
	\item $\widehat{\bm{\theta}}^{\rm LS}$ is unbiased, namely, $\mathbb{E}\big( \widehat{\bm{\theta}}^{\rm LS}\big)=\bm{\theta}$;
	\item The mean squared error (MSE) matrix of $\widehat{\bm{\theta}}^{\rm LS}$ is
	 \begin{align}
	\mathbb{MSE }\big(\widehat{\bm{\theta}}^{\rm LS}\big)
	&\eq \mathbb{E} (\widehat{\bm{\theta}}^{\rm LS}  - \bm{\theta})(\widehat{\bm{\theta}}^{\rm LS}  - \bm{\theta})^\top\nonumber\\
	&=(\mathbf{A}^\top  \mathbf{A})^{-1}  \mathbf{A}^\top \mathrm{P} \mathbf{A} (\mathbf{A}^\top  \mathbf{A})^{-1}\label{ls.mse} 
	\end{align} 
	where $\mathrm{P}=\diag\big([p_1-p_1^2,\cdots,p_M-p_M^2]\big)/n$.
\end{itemize}
However, this standard least squares neglected the fact that the $e_m$ have different variances, although they are all zero mean. As a result, the above covariance is not optimal. Furthermore, the condition that $\mathbf{A}$ be full column rank means the POVM $\{\mathsf{M}_m \}_{m=1}^M$  is informationally complete, i.e., any two density operators are distinguishable under the POVM given sufficiently large number of samples. This is not practical for large-scale quantum systems. 

\medskip

\section{Regularized Linear Regressions} 
In this section, we present a few generalizations to standard LSE for the considered quantum state tomography problem, and investigate their performances in terms of mean-square errors.
\label{secregular}
\subsection{Weighted Regression}
Noticing  $\mathbb{V}(e_m)=(p_m-p_m^2)/n$, we can instead use the following weighted least squares (WLS) estimate
\begin{align}
\widehat{\bm{\theta}}^{\rm WLS} &=\argmin_{\bm{\theta}}  (\mathbf{y}-\mathbf{A} \bm{\theta})^\top  \mathrm{W} (\mathbf{y}-\mathbf{A} \bm{\theta})\nonumber\\
&=(\mathbf{A}^\top  \mathrm{W}\mathbf{A})^{-1} \mathbf{A}^\top \mathrm{W} \mathbf{y} \label{wls}
\end{align}
with $
\mathrm{W}=\mathrm{P}^{-1}=n\diag\big([1/(p_1-p_1^2),\cdots,1/(p_M-p_M^2)]\big)$ penalizing the difference in variances for the noises $e_m$. This weighted least square $\widehat{\bm{\theta}}^{\rm WLS} $ continues to be unbiased since $\mathbb{E}\big(\widehat{\bm{\theta}}^{\rm WLS}\big)=\bm{\theta} $ is easily verifiable
and its MSE is
\begin{align}\label{wls.mse}
\mathbb{MSE}\big( \widehat{\bm{\theta}}^{\rm WLS} \big)
=(\mathbf{A}^\top  \mathrm{W}\mathbf{A})^{-1}.
\end{align}
Suppose $\rank(\mathbf{A})=d^2$ and let $\widehat{\bm{\theta} }  $ be any linear unbiased estimate for $\bm{\theta}$.
Then we have
\begin{align*}
\mathbb{MSE}(\widehat{\bm{\theta}} ) \geq \mathbb{MSE}(\widehat{\bm{\theta}}^{\rm WLS} ).
\end{align*}
This means it is the best estimator of $\bm{\theta}$ among all unbiased linear estimators in the sense that it achieves the minimal covariance.

In practice, the matrix $\mathrm{W}$ in \eqref{wls} is unknown and a feasible solution is to use the estimate 
\begin{align}\label{awls}
\widehat{\bm{\theta}}^{\rm AWLS}=(\mathbf{A}^\top  \widehat{\mathrm{W}}\mathbf{A})^{-1} \mathbf{A}^\top \widehat{\mathrm{W}} \mathbf{y},
\end{align}
where $\mathrm{W}$ in \eqref{wls} is replaced by its consistent estimate
\begin{align}
\label{west}
\widehat{\mathrm{W}}=n\cdot\diag\big([1/(\widehat{p}_1-\widehat{p}_1^2),\cdots,1/(\widehat{p}_M-\widehat{p}_M^2)]\big)
\end{align}
with $\widehat{p}_m,1\leq m \leq M$ given by \eqref{pm_est}.
In the following, it is shown that the estimate \eqref{awls} is accurate enough and asymptotically coincides with \eqref{wls}.

	For	a random sequence $\xi_n$, 
	we define $\xi_n=O_p(a_n)$ by that 
	$\{\xi_n/a_n\}$ is bounded in probability, i.e., $\forall
	\epsilon>0,\exists L>0$ such that $\mathbb{P}(|\xi_n/a_n|>L)<\epsilon,~\forall
	n$. Then there holds for large $n$ that
	\begingroup
\allowdisplaybreaks
\begin{align}
\widehat{\bm{\theta}}^{\rm AWLS}
-\widehat{\bm{\theta}}^{\rm WLS}
&= (\mathbf{A}^\top  \widehat{\mathrm{W}}\mathbf{A})^{-1} \mathbf{A}^\top  \widehat{\mathrm{W}}\mathbf{e}
-(\mathbf{A}^\top  \mathrm{W}\mathbf{A})^{-1} \mathbf{A}^\top  \mathrm{W}\mathbf{e}\\
\nonumber
&=\big( (\mathbf{A}^\top  \widehat{\mathrm{W}} \mathbf{A})^{-1} 
- (\mathbf{A}^\top \mathrm{W} \mathbf{A})^{-1} \big)\mathbf{A}^\top  \widehat{\mathrm{W}}\mathbf{e}
+(\mathbf{A}^\top  \mathrm{W}\mathbf{A})^{-1} \mathbf{A}^\top  \big(  \widehat{\mathrm{W}} -\mathrm{W}\big)\mathbf{e}\\
\nonumber
&=O_p(1/\sqrt{n}) (\mathbf{A}^\top \mathrm{W} \mathbf{A})^{-1} 
\mathbf{A}^\top  \mathrm{W} \big(1+ O_p(1/\sqrt{n})\big) \mathbf{e}
+(\mathbf{A}^\top  \mathrm{W}\mathbf{A})^{-1} \mathbf{A}^\top  O_p\big(1/\sqrt{n}\big) \mathrm{W} \mathbf{e}\\
&=O_p(1/\sqrt{n})(\mathbf{A}^\top  \mathrm{W}\mathbf{A})^{-1} 
\mathbf{A}^\top  \mathrm{W} \mathbf{e} \label{approx}
\end{align}
\endgroup
in terms of
\begin{align*}
e_m=O_p(1/\sqrt{n}),~1\leq m \leq M
\end{align*} 
and further
\begingroup
\allowdisplaybreaks
	\begin{align*}
	&\widehat{\mathrm{W}}
	=\mathrm{W} \big(1+ O_p(1/\sqrt{n})\big)\\
	&\widehat{\mathrm{W}} - \mathrm{W} = O_p\big(1/\sqrt{n}\big) \mathrm{W}.
	\end{align*}
\endgroup
This means $\widehat{\bm{\theta}}^{\rm AWLS}$ is a consistently practical approximation for the weighted LSE $\widehat{\bm{\theta}}^{\rm WLS}$.
Actually, the approximation \dref{west} and resulting conclusions \dref{approx} also hold for the following introduced estimators. 

\subsection{Constrained Weighted   Regression}
The standard or weighted least squares solutions might lead to estimates that are not legitimate quantum states. In fact, the quantum state has an essential requirement
\begin{align}
{\rm Tr} (\rho ) =1.
\end{align}
This becomes for the model \eqref{rho} that
\begin{align}
\bm{\theta}^\top {\rm Tr}(\mathsf{B})=1
\end{align}
where ${\rm Tr}(\mathsf{B})$ is defined by
\begin{align}
{\rm Tr}(\mathsf{B}) \eq [{\rm Tr}(\mathsf{B}_1),\cdots, {\rm Tr}(\mathsf{B}_{d^2}) ]^\top.
\end{align}
This inspires us to define  the constrained  least squares (CLS) estimate   
\begin{align}
\widehat{\bm{\theta}}^{\rm CLS}=\argmin_{\bm{\theta}^\top \tr(\mathsf{B})=1} \ ~(\mathbf{y}-\mathbf{A} \bm{\theta})^\top  (\mathbf{y}-\mathbf{A} \bm{\theta}). \label{cls}
\end{align}
For the estimate \eqref{cls}, we have the following proposition to characterize its property.

\begin{prop}
	\label{prop3}
	Suppose $\rank(\mathbf{A})=d^2$. The CLS estimate $\widehat{\bm{\theta}}^{\rm CLS}$ has the following closed-form solution
\begin{align}
\widehat{\bm{\theta}}^{\rm CLS} 
= \widehat{\bm{\theta}}^{\rm LS} 
-
\frac{C   \tr(\mathsf{B}) }
{ \tr(\mathsf{B})^\top  C  \tr(\mathsf{B})}
\big(\tr(\mathsf{B})^\top  \widehat{\bm{\theta}}^{\rm LS}  -1\big) \label{clses}
\end{align}
where $\widehat{\bm{\theta}}^{\rm LS} $ is the least squares estimate given by \eqref{ls} and $C=(\mathbf{A}^\top  \mathbf{A})^{-1}$,
and its MSE matrix is 
	 \begin{align*}
\mathbb{MSE }\big(\widehat{\bm{\theta}}^{\rm CLS} \big)
&\eq \mathbb{E} (\widehat{\bm{\theta}}^{\rm CLS}
  - \bm{\theta})
  (\widehat{\bm{\theta}}^{\rm CLS}  - \bm{\theta})^\top\\
&=F \mathbf{A} ^\top \mathrm{W}^{-1} \mathbf{A} F
\end{align*} 
where
$
F\eq C - \frac{C \tr(\mathsf{B}) \tr(\mathsf{B}) ^\top C}
{ \tr(\mathsf{B})^\top  C \tr(\mathsf{B})}. $
\end{prop}
To make the notation simple, we will a little abuse the symbols $F$ and $C$ for different cases in the following.

 To reduce the MSE of the estimate \eqref{cls}, we can similarly introduce 
 the constrained  weighted least squares (CWLS) estimate   
 \begin{align}\
 \widehat{\bm{\theta}}^{\rm CWLS}
 =\argmin_{\bm{\theta}^\top \tr(\mathsf{B})=1}
   ~(\mathbf{y}-\mathbf{A} \bm{\theta})^\top  \mathrm{W} (\mathbf{y}-\mathbf{A} \bm{\theta}). 
   \label{cwls}
 \end{align}

 \begin{thm}
 	\label{thm1}
 	Suppose $\rank(\mathbf{A})=d^2$ and $p_m\in(0,1)$ for $m=1,\cdots,M$.
 	The estimate $\widehat{\bm{\theta}}^{\rm CWLS}$ can be explicitly written as 
$$
\widehat{\bm{\theta}}^{\rm CWLS} 
= \widehat{\bm{\theta}}^{\rm WLS} 
-
\frac{C \tr(\mathsf{B}) }
{ \tr(\mathsf{B})^\top  C \tr(\mathsf{B})}
\big(\tr(\mathsf{B})^\top  \widehat{\bm{\theta}}^{\rm WLS}  -1\big)
$$
where $\widehat{\bm{\theta}}^{\rm WLS}$ is the WLS estimate \eqref{wls} and $C=(\mathbf{A}^\top  \mathrm{W} \mathbf{A})^{-1} $. The resulting MSE
\begin{align}\label{cwls.mse}
\mathbb{MSE}(\widehat{\bm{\theta}}^{\rm CWLS} )
=\mathbb{E} (\widehat{\bm{\theta}}^{\rm CWLS}  - \bm{\theta})
	(\widehat{\bm{\theta}}^{\rm CWLS}  - \bm{\theta})^\top
	=F 
\end{align}
where
$
F\eq C - \frac{C \tr(\mathsf{B}) \tr(\mathsf{B}) ^\top C}
{ \tr(\mathsf{B})^\top  C \tr(\mathsf{B})} $,
	is optimal in the sense that
	\begin{align*}
	\mathbb{MSE}(\widehat{\bm{\theta}} ) \geq \mathbb{MSE}(\widehat{\bm{\theta}}^{\rm CWLS} )
	\end{align*}
	where $\widehat{\bm{\theta}}$ is any unbiased estimate for $\bm{\theta}$ that is affine in $\mathbf{y}$ and $\bm{\theta}$ satisfies the constraint
	$\bm{\theta}^\top \tr(\mathsf{B})=1$.
	\end{thm}

\subsection{Regularized Weighted Regression}
Further, we introduce the following weighted regression with $\ell_2$-regularization: 
\begin{subequations}
	\label{rop}
\begin{align}
&\mathop{\rm minimize} \limits_{\bm{\theta}} \quad  ~(\mathbf{y}-\mathbf{A} \bm{\theta})^\top  \mathrm{W} (\mathbf{y}-\mathbf{A} \bm{\theta})+\gamma\|\bm{\theta}\|^2\\
&\mbox{subject to} \quad ~\bm{\theta}^\top \tr(\mathsf{B})=1.\label{cons2}
\end{align}
\end{subequations}
where $\gamma\geq 0$ is a regularization parameter
and $\|\cdot\|$ represents the 2-norm of a vector. The motivation for introducing (\ref{rop}) may arise  from the following two aspects:
\begin{itemize}
\item[(i)] When the POVM $\{\mathsf{M}_m \}_{m=1}^M$ is under-determinate, the  matrix $\mathbf{A}$ in (\ref{rg}) might not have full column rank. As a result, the $\widehat{\bm{\theta}}^{\rm LS}$, $ \widehat{\bm{\theta}}^{\rm CLS}$, and $\widehat{\bm{\theta}}^{\rm CWLS}$ will all fail to produce a unique estimate to the quantum state. The additional $\ell_2$-regularization term in the prediction error function will resolve this non-uniqueness challenge. 
\item[(ii)] In practice, the quantum state $\rho$ is often a combination of some finite number of pure states. As a result, a significant prior knowledge on  $\rho$ would be that it  is of low rank.
Since the rank minimization optimization problem with convex constraints is NP-hard \citep{Recht2010}, the nuclear norm is a common alternative as an approximation of the rank constraint for matrices in various matrix optimization problems. Note that $\rho^\dagger\rho$ has the same rank as that of $\rho$.
As a result,  $\rho^\dagger\rho$ is still of low rank
and the nuclear norm of $\rho^\dagger\rho$ is
	\begingroup
\allowdisplaybreaks
\begin{align*}
\|\rho^\dagger\rho\|_\star &\eq \sum_{i=1}^d \sigma_i(\rho^\dagger\rho)=\tr(\rho^\dagger\rho) \\
&= 
\tr\left[\left( \sum_{i=1}^{d^2}  \theta_i \mathsf{B}_i  \right)^\dagger 
\left( \sum_{j=1}^{d^2}  \theta_j \mathsf{B}_j  \right) \right]\\
&=\sum_{i=1}^{d^2} |\theta_i|^2
\\& = \|\bm{\theta}\|^2.
\end{align*}
\endgroup
Therefore $\ell_2$ regularization can be a good rank penalty as well. 
\end{itemize}
The two aspects are certainly connected in practice, where reconstruction of unknown low-rank quantum state is desried with a small number of measurement basis.

\begin{rem}
For the positive semidefinite quantum state $\rho$,   penalizing the nuclear norm $\rho$ (see, e.g.,  \cite{gross2010}) is not quite well-defined because
	\begingroup
\allowdisplaybreaks
\begin{align}
\nonumber
\|\rho\|_\star &\eq \sum_{i=1}^{d} \sigma_i(\rho) 
=\sum_{i=1}^{d} \sqrt{\lambda_i(\rho^\dagger\rho) }\\
&=\sum_{i=1}^{d} \lambda_i(\rho) 
={\rm Tr}(\rho) =1,
\end{align}
\endgroup
where $\sigma_i$ and $\lambda_i$ are the singular values and eigenvalues of $\rho$, respectively. Note that (\ref{rop}) is essentially the regularized optimization approach adopted in \cite{gross2010} for the numerical study of quantum state reconstruction problems. 
\end{rem}

\begin{rem} The optimization problem (\ref{rop}) can be equivalently represented as 
\begin{subequations}
	\begin{align}
	&\mathop{\rm minimize} \limits_{\bm{\theta}} \quad  ~(\mathbf{y}-\mathbf{A} \bm{\theta})^\top  \mathrm{W} (\mathbf{y}-\mathbf{A} \bm{\theta})\\
	&\mbox{subject to} \quad ~\bm{\theta}^\top \tr(\mathsf{B})=1,~\|\bm{\theta}\|^2\leq c\label{cons1}
	\end{align}
\end{subequations}
where $c>0$ corresponds to $\gamma$. In (\ref{cons1}), it is clear that the $\ell_2$ norm of the $\bm{\theta}$ serves as a constraint from the two aspects of motivations for such regularization. 
\end{rem}

For convenience and consistence of the results displayed in the paper, here we first introduce the 
regularized weighted least squares (RWLS) estimate
\begin{subequations}
	\begin{align}
\widehat{\bm{\theta}}^{\rm RWLS} 
&\eq
\argmin_{\bm{\theta}} 
(\mathbf{y}\!-\!\mathbf{A} \bm{\theta})^\top  \mathrm{W} (\mathbf{y} \!-\!\mathbf{A} \bm{\theta})+\gamma\|\bm{\theta}\|^2\\
&=(\mathbf{A}^\top  \mathrm{W}\mathbf{A} + \gamma \mathsf{I})^{-1}\mathbf{A} ^\top \mathrm{W} \mathbf{y},
\end{align}
\end{subequations}
where the constraint $\bm{\theta}^\top \tr(\mathsf{B})=1$ is neglected. 

The problem \eqref{rop} also has  a closed-from solution, which is stated in the following theorem.
\begin{thm}
	\label{thm2}
The optimal weighted  regularized quantum state estimate, denoted $\widehat{\bm{\theta}}^{\rm CRWLS}$,  as the solution to \eqref{rop} is given by
	\begin{align}
	&\widehat{\bm{\theta}}^{\rm CRWLS}  
	= \widehat{\bm{\theta}}^{\rm RWLS} 
	-C
	\tr(\mathsf{B})
	\frac{ \tr(\mathsf{B})^\top  \widehat{\bm{\theta}}^{\rm RWLS}  -1}{ \tr(\mathsf{B})^\top  C \tr(\mathsf{B})}
	\label{crwls}
	\end{align}
	where $C=(\mathbf{A}^\top  \mathrm{W}\mathbf{A} + \gamma \mathsf{I})^{-1}$.
	The resulting  MSE matrix of $\widehat{\bm{\theta}}^{\rm CRWLS}$ is
	\begin{align}
	\nonumber
	\mathbb{MSE}(\widehat{\bm{\theta}}^{\rm CRWLS})
	&\eq
	\mathbb{E} (\widehat{\bm{\theta}}^{\rm CRWLS}  - \bm{\theta})
	(\widehat{\bm{\theta}}^{\rm CRWLS}  - \bm{\theta})^\top\\
	&=F-\gamma F(  \mathsf{I}  - \gamma  \bm{\theta}\bm{\theta} ^\top)F
\label{mse1}
	\end{align}
	where  $F= C - \frac{C \tr(\mathsf{B}) \tr(\mathsf{B}) ^\top C}
	{ \tr(\mathsf{B})^\top  C \tr(\mathsf{B})} $. 
	\end{thm}


It is worth noting that Theorem~\ref{thm2} does not depend on the rank of $\mathbf{A}$. 
The next theorem shows that the CRWLS estimate $\widehat{\bm{\theta}}^{\rm CRWLS}$ yields immediate improvement in terms of mean squred error if the regularization parameter $\gamma$ is well chosen.

\begin{thm}
	\label{thm5}
	There holds
	$$\mathbb{MSE}(\widehat{\bm{\theta}}^{\rm CRWLS}) 
	<	\mathbb{MSE}(\widehat{\bm{\theta}}^{\rm CWLS})$$
	if
	$0<\gamma<2/\big(\|\bm{\theta}\|^2
	-\frac1{\|{\rm Tr}(\mathsf{B})\|^2}\big)$.
\end{thm}	
\begin{rem}
	\label{rem3}
There holds from the Cauchy–Schwarz inequality that  $$
\|\bm{\theta}\|^2 \|{\rm Tr}(\mathsf{B})\|^2 \geq |\bm{\theta}^\top {\rm Tr}(\mathsf{B})|^2 =  1
$$ for all  quantum states $\rho$. Moreover, when strict equality takes place, there is $\lambda \in \mathbb{R}$ such that  $\theta_i= \lambda {\rm Tr} (\mathsf{B}_i)$ for all $i=1,\dots,d^2$. As a result, 
$$
 {\rm Tr} (\rho \mathsf{B}_i)= \theta_i= \lambda {\rm Tr} (\mathsf{B}_i), \ \ i=1,\dots,d^2
$$
which implies $\rho =\lambda \mathsf{I}$, and  hence $\lambda$ must be $1/d$. Therefore, we have just established  that$$
\|\bm{\theta}\|^2>\frac1{\|{\rm Tr}(\mathsf{B})\|^2}
	$$
for all $\rho$ as quantum states except for $\rho = \mathsf{I}/d$.
\end{rem}	
Theorem \ref{thm5} shows that regularization that considers the low rank property of the quantum state $\rho$ can further improve the estimate for the parameter vector $\bm{\theta}$ if we can choose a proper $\gamma$.

\begin{rem}
Theorem \ref{thm1} indicates $\widehat{\bm{\theta}}^{\rm CWLS}$ has the smallest MSE among all the unbiased estimate of $\bm{\theta}$ linear with $\mathbf{y}$ while Theorem \ref{thm5} shows that $\widehat{\bm{\theta}}^{\rm CRWLS}$ has a smaller MSE than $\widehat{\bm{\theta}}^{\rm CWLS}$ even if $\widehat{\bm{\theta}}^{\rm CRWLS}$ is also linear with $\mathbf{y}$.
	 The reason is that regularization introduces a small bias but decreases the variance more such that the total MSE is smaller.
\end{rem}	
 
The estimate $\widehat{\bm{\theta}}^{\rm CRWLS}$ is a function of the regularization parameter $\gamma$, the selection of which needs to be determined carefully to achieve a better performance.
The essence of tuning $\gamma$ is to choose a proper model complexity for the estimate $\widehat{\bm{\theta}}^{\rm CRWLS}$  given the data.
Here we provide a method of tuning $\gamma$ by the measurements based on the risk definition of the estimate $\widehat{\bm{\theta}}^{\rm CRWLS}$.
Also, we will prove that the tuning method is asymptotically optimal in the risk sense. For convenience of derivation, rewrite the estimate $\widehat{\bm{\theta}}^{\rm CRWLS}$
as the affine form with respective the output $\mathbf{y}$
\begin{align}
\widehat{\bm{\theta}}^{\rm CRWLS}=\mathbf{H}\mathbf{y} + \bm{f},
\end{align}
where with $C=(\mathbf{A}^\top  \mathrm{W}\mathbf{A} + \gamma \mathsf{I})^{-1}$,
	\begingroup
\allowdisplaybreaks
\begin{align*}
&\mathbf{H}=C\mathbf{A}^\top  \mathrm{W} 
-C \tr(\mathsf{B}) \nonumber
	\frac{ \tr(\mathsf{B})^\top  C \mathbf{A}^\top  \mathrm{W} }{ \tr(\mathsf{B})^\top  C\tr(\mathsf{B})},\\
	&\bm{f}=\frac{C \tr(\mathsf{B})}
	{ \tr(\mathsf{B})^\top  C \tr(\mathsf{B})}.
\end{align*}
\endgroup
Let us introduce the risk for the estimate $\widehat{\bm{\theta}}^{\rm CRWLS}$ defined by \citep{Rao2018}:
\begin{subequations}
	\label{thetarisk}
	\begin{align}
R(\widehat{\bm{\theta}}^{\rm CRWLS} )
&
\eq \mathbb{E}\big(\mathbf{A} \bm{\theta} -   \mathbf{A}\widehat{\bm{\theta}}^{\rm CRWLS}  \big)^\top  
\mathrm{W}
\big(\mathbf{A} \bm{\theta} -   \mathbf{A}\widehat{\bm{\theta}}^{\rm CRWLS}  \big)\\\
&=\gamma^2
\bm{\theta} ^\top
\! F
\mathbf{A}^\top  
\!
\mathrm{W}\mathbf{A} 
F\bm{\theta}
+\tr\left(  F  \mathbf{A}^\top  \mathrm{W}\mathbf{A}  
F \mathbf{A}^\top 
\!\mathrm{W} \mathbf{A}  \right)
\label{risktheta}
\end{align}
\end{subequations}
which is a reference measure to characterize how well the estimate \eqref{crwls} can achieve, namely, 
gives an upper bound of the estimate \eqref{crwls} in the risk sense \eqref{thetarisk}.
Thus the regularization parameter $\gamma$  tuned by the risk \eqref{thetarisk}
\begin{align}
\widehat{ \gamma}_{R}
(\widehat{\bm{\theta}}^{\rm CRWLS} )
\eq\argmin_{\gamma \geq 0}R(\widehat{\bm{\theta}}^{\rm CRWLS} )
\label{rprisktheta}
\end{align}
is the optimal regularization parameter of $\gamma$ for any given data in the  risk sense.
Unfortunately, the cost function \eqref{thetarisk} of \eqref{rprisktheta}  requires the access to the true parameter $\bm{\theta}$, which is usually unavailable for a system to be identified.

In the following, we use an unbiased estimate for \eqref{thetarisk} as the cost function of an implementable tuning estimator in terms of data to estimate $\gamma$,
which is given by
\begin{align}
\widehat{ \gamma}_u
(\widehat{\bm{\theta}}^{\rm CRWLS} )=\argmin_{\gamma\geq 0}\
(\mathbf{y} - \mathbf{A}\widehat{\bm{\theta}}^{\rm CRWLS})^\top
\mathrm{W}
(\mathbf{y} - \mathbf{A}\widehat{\bm{\theta}}^{\rm CRWLS})+2\tr\big( \mathbf{A} \mathbf{H}  \big). 
\label{unbiasedrisktheta}
\end{align}
It can be verified that the expectation of the cost function \dref{unbiasedrisktheta} over the estimation error $\mathbf{e}$ is exactly the risk \dref{thetarisk}.

The properties of $\widehat{ \gamma}_{R}
(\widehat{\bm{\theta}}^{\rm CRWLS} )$ and $\widehat{ \gamma}_u
(\widehat{\bm{\theta}}^{\rm CRWLS} )$
will be given in the next section under an alternative regression model.

\section{An Equivalent Regression Model}
\label{sec4}
Up to now our discussions on the quantum state tomography problem have been around the linear model with an equality constraint:
\begin{subequations}
	\label{constlm}
	\begin{align}
&\mathbf{y}=\mathbf{A}\bm{\theta} + \mathbf{e}\\
&\mbox{subject to }\bm{\theta}^\top \tr(\mathsf{B})=1.\label{consp}
\end{align}
\end{subequations}
In this section, we first present a way of transforming this standing linear regression model into an unconstrained version. Then, under the new but equivalent model we establish some important asymptotic  properties of the regularized regression solutions as sample size grows.

\subsection{Eliminating Equality  Constraint} 
Let us first construct an orthogonal matrix $\mathbf{Q}$ of size $d^2\times d^2$ as follows:
The first row of $\mathbf{Q}$ is ${\rm Tr}(\mathsf{B})^\top/\|{\rm Tr}(\mathsf{B})\|$ and the remaining rows are chosen such that $\mathbf{Q}$ is orthogonal.
Thus, we have from \dref{constlm} that
\begin{align}
\mathbf{y}=\mathbf{D}\bm{\beta} + \mathbf{e}  \label{equimodel}
\end{align}
where
\begin{subequations}
	\label{qmatrix}
	\begin{align}
&\mathbf{D}\eq \mathbf{A}\mathbf{Q}^\top=[\mathbf{d},\mathbf{K}]\\
&\bm{\beta} \eq \mathbf{Q}\bm{\theta}=[\beta_1,\bm{\alpha}^\top]^\top. \label{beta}
\end{align}
\end{subequations}
The constraint \eqref{consp} on $\bm{\theta}$ is forced by the fact that the first element $\beta_1$ of $\bm{\beta}$ is $1/\|{\rm Tr}(\mathsf{B})\|$.
As a result,  the problem \eqref{constlm} is equivalent to  the unconstrained linear model 
\begin{align}
\mathbf{z}=\mathbf{y} - \frac{1} {\|{\rm Tr}(\mathsf{B})\|}\mathbf{d}
= \mathbf{K} \bm{\alpha} + \mathbf{e}. 
\label{unconslm}
\end{align}
Clearly, by \eqref{beta}
\begin{align}
\|\bm{\alpha}\|^2 
=
\|\bm{\theta}\|^2
-\frac1{\|{\rm Tr}(\mathsf{B})\|^2}.\label{alphatheta}
\end{align}
Thus,  regularization (low rank property) on  $\rho^\dagger\rho$ 
can also be embedded into $\bm{\alpha}$.

For the model \eqref{unconslm}, 
we can produce the corresponding LS, WLS, RWLS estimates.
Here, we  consider the RWLS estimate for \eqref{unconslm} since other estimates (LS, WLS) are special cases by setting $\gamma=0$ and/or $\mathrm{W}=\mathsf{I}$.
The RWLS estimate for \eqref{unconslm} is defined as
	\begin{align}\label{unconstrwls}
	\widehat{\bm{\alpha}}^{\rm RWLS} 
	=
	\argmin_{\bm{\alpha}} 
	(\mathbf{z} \!-\! \mathbf{K} \bm{\alpha})^\top  \mathrm{W}
	(\mathbf{z}\!-\!\mathbf{K} \bm{\alpha}) \!+\!\gamma\|\bm{\alpha}\|^2
	=\mathbf{U}\mathbf{z}
	\end{align}
where
	\begin{align}
	\mathbf{U}\eq \mathbf{V}\mathbf{K} ^\top \mathrm{W},~~
	\mathbf{V}\eq (\mathbf{K}^\top  \mathrm{W}\mathbf{K} + \gamma \mathsf{I})^{-1}. 
	\end{align}
Intuitively, for an estimate $\widehat{\bm{\alpha} }$ of \eqref{unconslm},
the vector defined by
\begin{align}
\widehat{\bm{\theta} }(\widehat{\bm{\alpha} })\eq 
\mathbf{Q}^\top\left[
\begin{array}{c}
\frac1 {\|{\rm Tr}(\mathsf{B})\|}\\
\widehat{\bm{\alpha} }
\end{array}
\right]
\end{align}
should be the corresponding estimate for \eqref{constlm} and independent of the choice of $\mathbf{Q}$.
However, this is not obvious.
Now, we intend to show that the hypothesis above is true.

\begin{prop}
	\label{prop4}
For any regularization parameter $\gamma\geq 0$, there holds
	\begin{align} 
	\widehat{\bm{\theta} }(\widehat{\bm{\alpha}}^{\rm RWLS} )
	=\widehat{\bm{\theta}}^{\rm CRWLS}. \label{equiv}
	\end{align}
	Moreover,
	\begin{align*}
	\mathbb{MSE }\big(\widehat{\bm{\alpha}}^{\rm RWLS} (\gamma)\big)
	&\eq \mathbb{E} (\widehat{\bm{\alpha}}^{\rm RWLS}
	- \bm{\alpha})
	(\widehat{\bm{\alpha}}^{\rm RWLS}
	 - \bm{\alpha})^\top\\
	&=\gamma^2  \mathbf{V} \bm{\alpha} \bm{\alpha}^\top \mathbf{V}
	+ \mathbf{V} \mathbf{K}^\top 
	\mathrm{W}
	\mathbf{K} 
	\mathbf{V}.
	\end{align*} 
%

\end{prop}
\begin{rem}
	\label{rem4}
	When $\gamma=0$, the estimate $\widehat{\bm{\alpha}}^{\rm RWLS}$ is reduced to 
	the WLS estimate of \eqref{unconslm}.
	Meanwhile, we have
	\begin{align}
	\mathbb{MSE }\big(\widehat{\bm{\alpha}}^{\rm RWLS} (\gamma)\big)
	<
	\mathbb{MSE }\big(\widehat{\bm{\alpha}}^{\rm RWLS} (0)\big)
	\label{gg1}
	\end{align}
	for $0<\gamma<2/\bm{\alpha}^\top\bm{\alpha}$, an equivalent statement as Theorem \ref{thm5}. 
	\end{rem}
\subsection{Asymptotically Optimal Regularization Gain }
For the estimate \eqref{unconstrwls}, it also needs to well tune  the regularization parameter $\gamma$.
The risk for the estimate $\widehat{\bm{\alpha}}^{\rm RWLS}$ can be similarly defined as
	\begin{align}
	\nonumber
	&R(\widehat{\bm{\alpha}}^{\rm RWLS} )
	\eq \mathbb{E}\big(\mathbf{K} \bm{\alpha} -   \mathbf{K}\widehat{\bm{\alpha}}^{\rm RWLS}  \big)^{\!\!\!\top}  
	\mathrm{W}
	\big(\mathbf{K} \bm{\alpha} -   \mathbf{K}\widehat{\bm{\alpha}}^{\rm RWLS}  \big)\\
	&\hspace{2mm}=\gamma^2   \bm{\alpha}^\top \mathbf{V}
	\mathbf{K}^\top
	\mathrm{W} 
	\mathbf{K}
	\mathbf{V} \bm{\alpha}
	+ 
	\tr\big( 
	\mathbf{V} \mathbf{K}^\top 
	\mathrm{W}
	\mathbf{K} 
	\mathbf{V} 
	\mathbf{K}^\top
	\mathrm{W}
	\mathbf{K}  \big)\!\!\!
	\label{riskalpha}
	\end{align}
	and	the resulting optimal regularization parameter is
	\begin{align}
	\widehat{ \gamma}_{R}
	(\widehat{\bm{\alpha}}^{\rm RWLS} )
	\eq\argmin_{\gamma \geq 0}
	R(\widehat{\bm{\alpha}}^{\rm RWLS}(\gamma) ).
	\label{alpharisk}
	\end{align}
Let us construct an unbiased estimate
\begin{align}
\mathscr{C}_u(\gamma)\eq(\mathbf{z} - \mathbf{K}\widehat{\bm{\alpha}}^{\rm RWLS})^\top
\mathrm{W}
(\mathbf{z} - \mathbf{K}\widehat{\bm{\alpha}}^{\rm RWLS})
+2\tr\big( \mathbf{K} \mathbf{U}  \big)
\label{cu}
\end{align}
 for \eqref{riskalpha} and it can straightforwardly check its expectation with respect to $\mathbf{e}$ is $R(\widehat{\bm{\alpha}}^{\rm RWLS} )$ up to a constant.
Consequently, we propose the tuning estimator for $\gamma$ as
	\begin{align}
	\widehat{ \gamma}_u
	(\widehat{\bm{\alpha}}^{\rm RWLS} )
	\eq\argmin_{\gamma\geq 0}\
\mathscr{C}_u(\gamma)
	\label{unbiasedriskalpha} 
	\end{align}
which gives a way to estimate $\gamma$ directly by the data.

The following proposition illustrates that the tuning estimators \dref{rprisktheta} and \dref{unbiasedrisktheta} as well as \dref{alpharisk} and \dref{unbiasedriskalpha} 
developed for the constrained model \eqref{unbiasedrisktheta} and its unconstrained counterpart \eqref{unbiasedriskalpha}, respectively,
are equivalent.
\begin{prop}\label{prop5}

There hold
	\begingroup
\allowdisplaybreaks
\begin{subequations}
		\begin{align}
	\widehat{ \gamma}_{R}
	(\widehat{\bm{\alpha}}^{\rm RWLS} )
	&=
	\widehat{ \gamma}_{R}
	(\widehat{\bm{\theta}}^{\rm CRWLS} )\label{eqhp2}\\
	 \widehat{ \gamma}_u
	(\widehat{\bm{\alpha}}^{\rm RWLS} )
	&=
	\widehat{ \gamma}_u
	(\widehat{\bm{\theta}}^{\rm CRWLS} ).\label{eqhp4}
	\end{align}
\end{subequations}
\endgroup
\end{prop}

Denote
\begin{subequations}
	\begin{align}
&\Sigma \eq 
\mathbf{K}^\top
\diag\big([p_1-p_1^2,\cdots,p_M-p_M^2]\big)
\mathbf{K}\\
&\Upsilon\eq\mathbf{A}^\top
\diag\big([p_1-p_1^2,\cdots,p_M-p_M^2]\big)
\mathbf{A} .
\end{align}
\end{subequations}
We can establish the asymptotically optimal selection of the regularization parameter $\gamma$ explicitly for the regularized regression estimate  of the quantum state in the risk senses \dref{thetarisk} and \dref{riskalpha}. 
\begin{thm}\label{thm6}
	Suppose $\rank(\mathbf{A})=d^2$.
The limits take place as the sample size $n\xra{}\infty$ by
	\begin{subequations}
		\label{ophp}
		\begin{align}
		&\widehat{ \gamma}_{R}
		(\widehat{\bm{\alpha}}^{\rm RWLS} )
		\xra{} \gamma^\star~ \mbox{deterministically}\label{ophp1} \\
		&\widehat{ \gamma}_u
		(\widehat{\bm{\alpha}}^{\rm RWLS} ) 
		\xra{} \gamma^\star~ \mbox{almost surely}\label{ophp2} 
		\end{align}
	\end{subequations} 
	where
		\label{asymrp}
		\begin{align*}
		\nonumber
		\gamma^\star
		=\frac{\tr\big(  \Sigma^{-1} \big)}
		{\bm{\alpha} ^\top \Sigma^{-1}\bm{\alpha} } 
		=\frac{\tr\big(\Upsilon^{-1}\big)
			-\frac{\tr(\mathsf{B})^\top \Upsilon^{-2} \tr(\mathsf{B}) }{ \tr(\mathsf{B})^\top  \Upsilon^{-1}\tr(\mathsf{B})}}{\bm{\theta}^\top 
			\Upsilon^{-1}
			\bm{\theta}
			-
			\frac{\bm{\theta} ^\top \Upsilon^{-1} \tr(\mathsf{B})  \tr(\mathsf{B})^\top \Upsilon^{-1}  \bm{\theta}}
			{ \tr(\mathsf{B})^\top  \Upsilon^{-1}\tr(\mathsf{B})}}
		\end{align*}
		is the asymptotically optimal selection of  $\gamma$  for the  estimate \dref{crwls} of the quantum state in the risk senses \dref{thetarisk} and \dref{riskalpha}.
Moreover, there hold  as $n\xra{}\infty$
	\begin{align}
n\big(\widehat{ \gamma}_{R}
(\widehat{\bm{\alpha}}^{\rm RWLS} )
- \gamma^\star\big)\xra{}~
\frac{3\gamma^\star
	\big(\gamma^\star
	\bm{\alpha}^\top
	\mathbf{\Sigma^{-2}}
	\bm{\alpha}
	-
	\tr\big( \Sigma^{-2} \big)
	\big)}
{\bm{\alpha}^\top
	\mathbf{\Sigma^{-1}}
	\bm{\alpha}}
\label{ophpd} 
\end{align}
deterministically and
\begin{align}
&\sqrt{n}
\big(	\widehat{ \gamma}_u 
(\widehat{\bm{\alpha}}^{\rm RWLS} ) 
-\gamma^\star\big)
\xra{}\mathscr{N}
\left(0, \frac{4(\gamma^\star)^2\bm{\alpha}^\top
	\Sigma^{-3}
	\bm{\alpha} }
{\big(\bm{\alpha}^\top
	\Sigma^{-1}
	\bm{\alpha}\big)^2 }
\right)
\label{ophpr} 
\end{align}
in distribution. 

\end{thm}
  
 \begin{rem}
 	\label{rem6}
 	 Theorem \ref{thm6} shows that the implementable estimators
 	 $
 	 \widehat{ \gamma}_u
 	 (\widehat{\bm{\theta}}^{\rm CRWLS} )$ 
 	  and
 	 $\widehat{ \gamma}_u
 	 (\widehat{\bm{\alpha}}^{\rm RWLS} )$
 	  converges to the asymptotically optimal $\gamma^\star$ as the optimal estimators $\widehat{ \gamma}_{R}
 	  (\widehat{\bm{\theta}}^{\rm CRWLS} )$ and $\widehat{ \gamma}_{R}
 	  (\widehat{\bm{\alpha}}^{\rm RWLS} )$ for any finite sample data do. 
 	  On the other hand, $
 	 \widehat{ \gamma}_u
 	 (\widehat{\bm{\theta}}^{\rm CRWLS} )$ 
 	 and
 	 $\widehat{ \gamma}_u
 	 (\widehat{\bm{\alpha}}^{\rm RWLS} )$ have a slower rate of convergence than that of $\widehat{ \gamma}_{R}
 	 (\widehat{\bm{\theta}}^{\rm CRWLS} )$ and $\widehat{ \gamma}_{R}
 	 (\widehat{\bm{\alpha}}^{\rm RWLS} )$.
 \end{rem}

\section{Numerical Examples}\label{sec5}
\subsection{Overdeterminate Measurement Basis}
\begin{exam}\label{exam1}\normalfont
	We consider  the following quantum Werner state tomography for a two-qubit system as studied in~\cite{qi2013}: 
	$$\rho_q=q\braket{{\boldmath \Psi}^-}{{\boldmath \Psi}^-}+\frac{1-q}{4}\mathsf{I}
	$$
	where $\bra{{\boldmath \Psi}^-}=({\bra{01}-\bra{10}})/{\sqrt{2}}$ and $q\in[0,1]$ is a parameter associated with the state.
	Take an orthonormal basis $\big\{\mathsf{B}_i\big\}_{i=1}^{16}$ as 
	$$
	\mathsf{B}_i=\frac{1}{\sqrt{2}}\sigma_{j}\otimes \frac{1}{\sqrt{2}}\sigma_{k},\quad i=4j+k+1
	$$
	for $j,k=0,1,2,3$, where \[\sigma_0=I_2,
	\sigma_1=\begin{bmatrix}0 & 1\\ 1 & 0\end{bmatrix},
	\sigma_2=\begin{bmatrix}0 & -i\\ i & 0\end{bmatrix},
	\sigma_3=\begin{bmatrix}1 & 0\\ 0 & -1\end{bmatrix}
	\] are the Pauli matrices.
	Let 
	\begin{align*}
	&\bra{\varphi_1}=\frac{1}{\sqrt{6}}[1,1]^\top, 
	\bra{\varphi_2}=\frac{1}{\sqrt{6}}[1,-1]^\top, 
	\bra{\varphi_3}=\frac{1}{\sqrt{6}}[1,i]^\top, \\
	&\bra{\varphi_4}=\frac{1}{\sqrt{6}}[1,-i]^\top, 
	\bra{\varphi_5}=\frac{1}{\sqrt{3}}[1,0]^\top, 
	\bra{\varphi_6}=\frac{1}{\sqrt{3}}[0,1]^\top.
	\end{align*} 
	
	Then
	$$
	\mathsf{E}_m=\braket{\varphi_{j}}{\varphi_{j}}\otimes \braket{\varphi_{k}}{\varphi_{k}},\quad m=6(j-1)+k,
	$$
	for $j,k=1,2,\dots,6$
	form our measurement basis $\big\{\mathsf{M}_m\big\}_{m=1}^{36}$ with $\mathsf{M}_m=\bra{\varphi_j}\otimes\bra{\varphi_k}$.
	We can verify that the measurement set $\big\{\mathsf{M}_m\big\}_{m=1}^{36}$ is overcomplete and the matrix $\mathbf{A}=[{\boldmath \beta}_1, \dots, {\boldmath \beta}_{36}]^\top\in\mathbb{R}^{36\times 16}$ has full column rank.
	
	We first sample the parameter $q$ to identify the quantum states $\rho_q$, and then for any $\rho_q$ carry out the tomography procedure for $1000$ rounds based on $n=110,1100,11000$ copies, respectively. The measurement process is simulated by i.i.d. random numbers $\lambda\in\{1,\dots,36\}$ according to $\mathbb{P}(\lambda=m)=p_m=\tr(\mathsf{E}_m\rho_q)$. Then for the experimental round $k$, $k=1,\dots,1000$, $\hat{p}_m{(k)}$ is recorded as the estimation of $p_m$. Based on the $\hat{p}_m{(k)}$, we derive the standard, weighted, and constrained weighted estimates $\widehat{\bm{\theta}}^{\mathrm{LS}}{(k)}$, $\widehat{\bm{\theta}}^{\mathrm{WLS}}{(k)}$, $\widehat{\bm{\theta}}^{\mathrm{CWLS}}{(k)}$,  and $\widehat{\bm{\theta}}^{\mathrm{CRWLS}}{(k)}$ (with $\gamma = 1/\big(\|\bm{\theta}\|^2
	-\frac1{\|{\rm Tr}(\mathsf{B})\|^2}\big)$) according to \eqref{ls}, \eqref{wls},  \eqref{cwls}, and \eqref{crwls} respectively, where $\mathrm{W}$ is replaced by its estimate \dref{west}. Note that some of the $\hat{p}_m$ might inevitably take value zero, and whenever that happens  we set $\hat{p}_m=10^{-8}$ in the weight matrix $\mathrm{W}$ for the sake of computation.   The experimental mean-square error $\mathbb{MSE}_{\rm exp}$ is  then computed by averaging the square error from each round of experiments, which are plotted with their theoretical predictions from \eqref{ls.mse}, \eqref{wls.mse}, \eqref{cwls.mse}, and \eqref{mse1},  for each state and their estimates 
	in  Figs.~\ref{fig:1}--\ref{fig:3}, where the true $\mathrm{W}$ instead of its estimate is used. 
	
	From these figures one easily sees that the experimental estimates are approaching  the theoretical ones as the number of samples $n$ grows large  for all four estimates, LS, WLS, CWLS, and CRWLS, which validates Theorem  \ref{thm1}, Theorem \ref{thm2}, and Theorem \ref{thm5}.  For small sample size ($n=110$), the WLS, CWLS, and CRWLS are apparently  producing worse experimental mean-square error compared to LS. The reason for that might  that  the $\widehat{\mathrm{W}}$ constructed from the $\hat{p}_m$  may not be accurate enough for approximating the true $\mathrm{W}$. For relatively larger sample size ($n=11000$), the WLS, CWLS, and CRWLS all provide significant improvments  compared to LS. It is worth noting that even with small sample size, the CRWLS may lead to drastically reduced error for small $q$, where $\rho_q$ tends to be closer to a separable quantum state.

	\begin{figure}[!htbp]
		\centering
		\includegraphics[width=.7\linewidth]{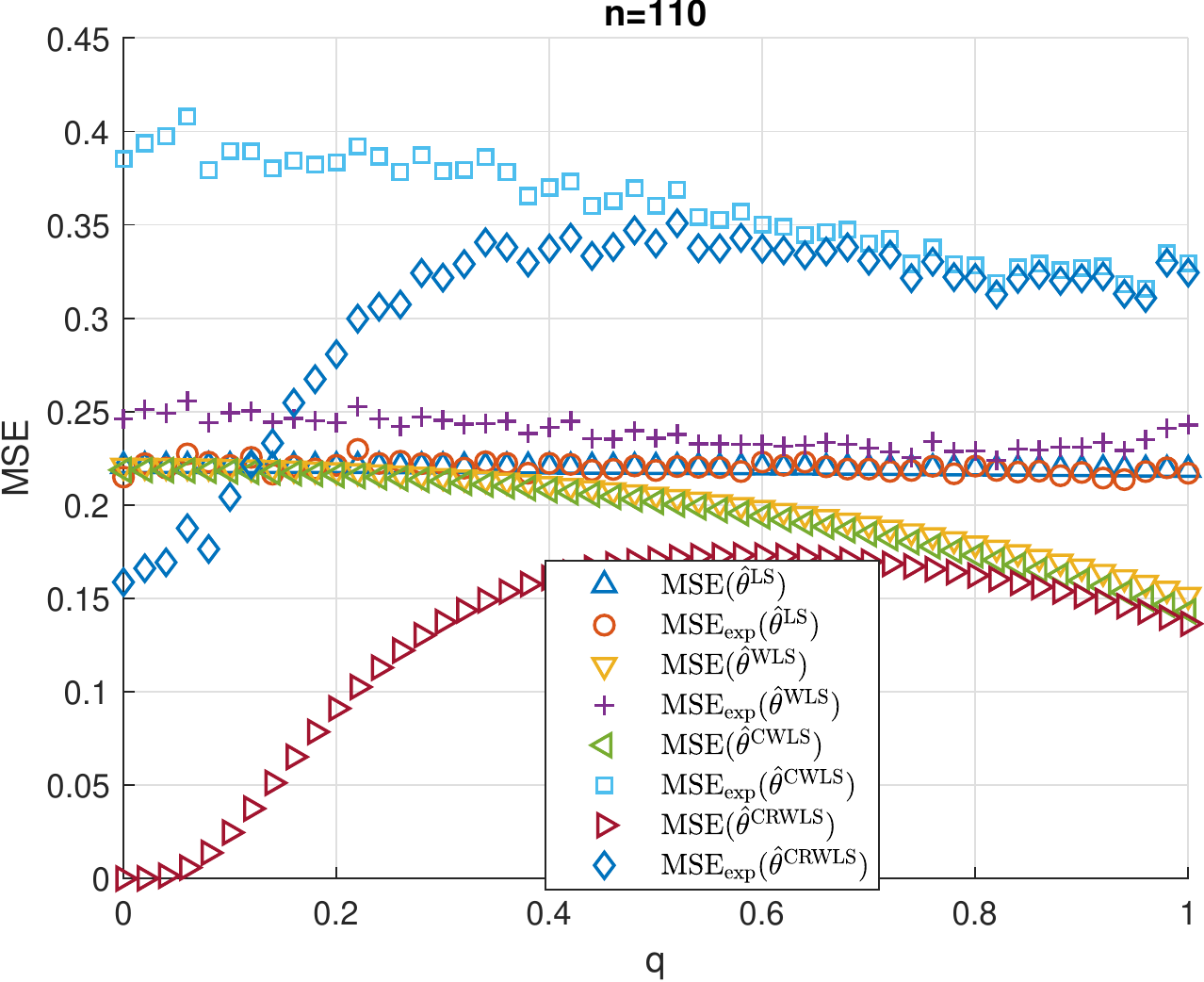}
		\caption{MSEs for estimating  Werner states with $n=110$ copies. 
		}\label{fig:1}
	\end{figure}
	
	\begin{figure}[!htbp]
		\centering
		\includegraphics[width=.7\linewidth]{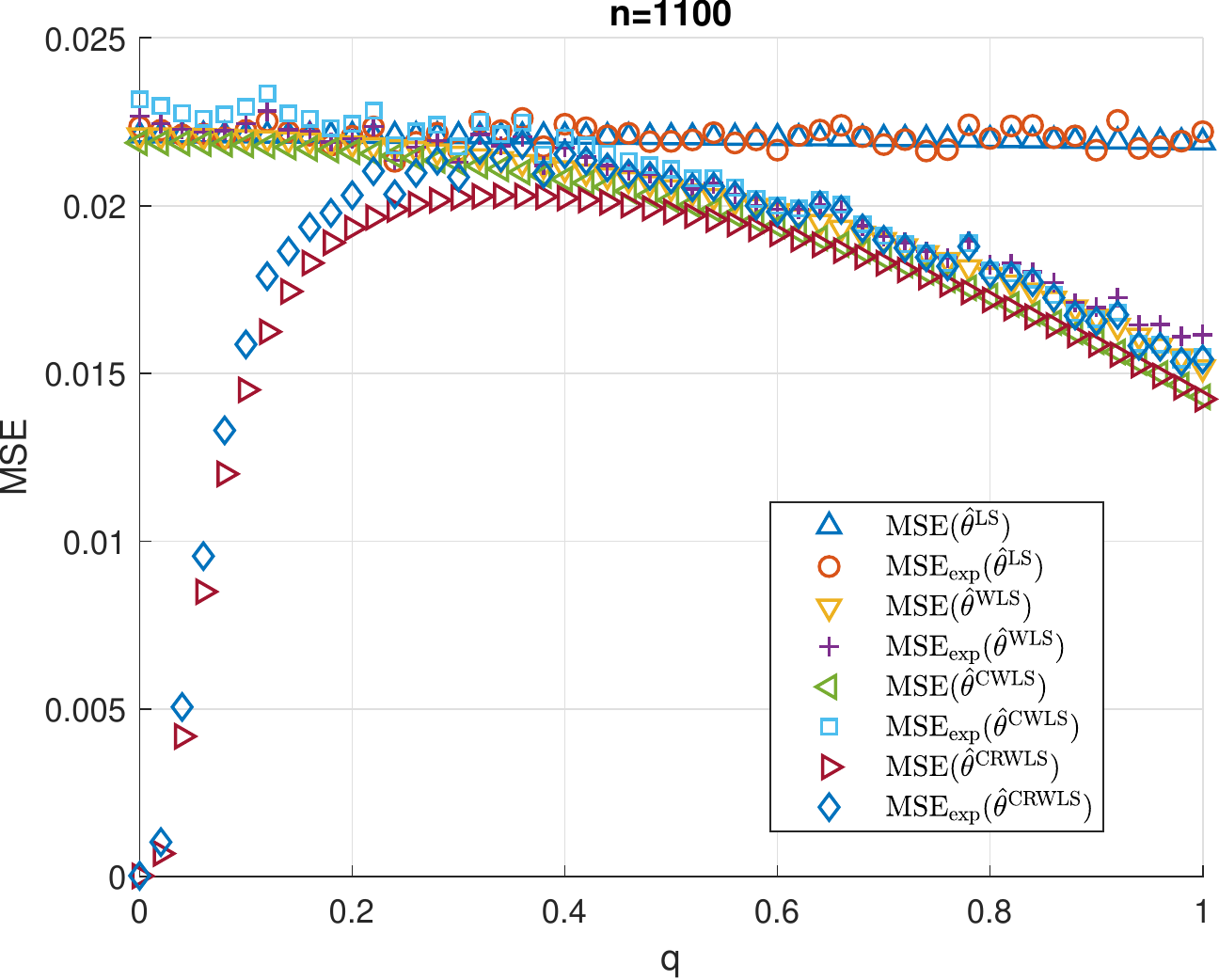}
		\caption{MSEs for estimating Werner states with $n=1100$ copies.}\label{fig:2}
	\end{figure}
	
	\begin{figure}[!htbp]
		\centering
		\includegraphics[width=.7\linewidth]{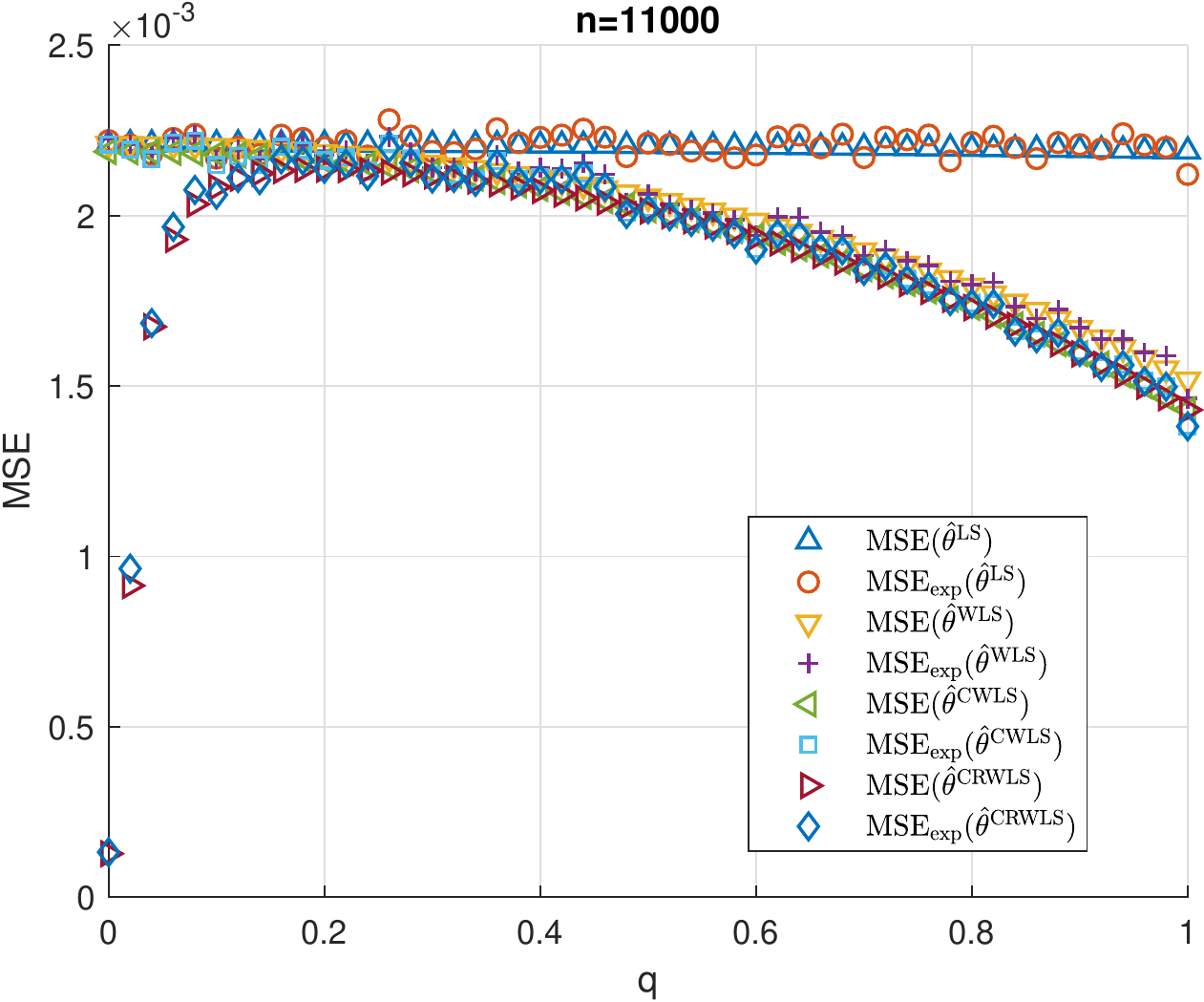}
		\caption{MSEs for estimating   Werner states with $n=11000$ copies.}\label{fig:3}
	\end{figure}
	
\end{exam}

\subsection{Small Sample-size and Optimal Regularizer}
As we have seen from Example \ref{exam1}, when the sample size $n$ is small, the weighted estimates  $\widehat{\bm{\theta}}^{\mathrm{WLS}}$, $\widehat{\bm{\theta}}^{\mathrm{CWLS}}$,  and $\widehat{\bm{\theta}}^{\mathrm{CRWLS}}$ may lead to lower accuracy compared to $\widehat{\bm{\theta}}^{\mathrm{LS}}$. In the following example, we show that in this case forcing $\mathrm{W}=\mathrm{I}$ in $\widehat{\bm{\theta}}^{\mathrm{CRWLS}}{(k)}$ to obtain a constrained regularized LS estimate (CRLS) would resolve the issue under the tuning methods in (\ref{rprisktheta}) and (\ref{unbiasedrisktheta}). 
\begin{exam} \normalfont
Consider exactly the same quantum state and tomography setup as in Example \ref{exam1}. Let $\mathrm{W}=\mathrm{I}$ in $\widehat{\bm{\theta}}^{\mathrm{CRWLS}} $ so that we define
$$
\widehat{\bm{\theta}}^{\mathrm{CRLS}} = \widehat{\bm{\theta}}^{\mathrm{CRWLS}} \big |_{\mathrm{W}=\mathrm{I}} 
$$
as the unweighted CRLS estimate. The regularization gain $\gamma$ is selected under the optimal value $\widehat{ \gamma}_{R}$ from (\ref{rprisktheta}) in the risk sense and its unbiased estimate $\widehat{ \gamma}_{u}$ from (\ref{unbiasedrisktheta}), under which for any $\rho_q$ we carry out the tomography procedure for $1000$ rounds based on $n=110,1100$ copies, respectively.   The resulting experimental mean-square errors $\mathbb{MSE}_{\rm exp}(\widehat{\bm{\theta}}^{\mathrm{CRLS}},\widehat{ \gamma}_{R}) $  and $\mathbb{MSE}_{\rm exp}(\widehat{\bm{\theta}}^{\mathrm{CRLS}},\widehat{ \gamma}_{u}) $ are  then computed and plotted in  Figs.~\ref{fig:r1}--\ref{fig:r2}, respectively, in comparison to  the theoretical and experimental MSE of standard LS estimate $\widehat{\bm{\theta}}^{\mathrm{LS}}$.

	\begin{figure}[!htbp]
		\centering
		\includegraphics[width=.7\linewidth]{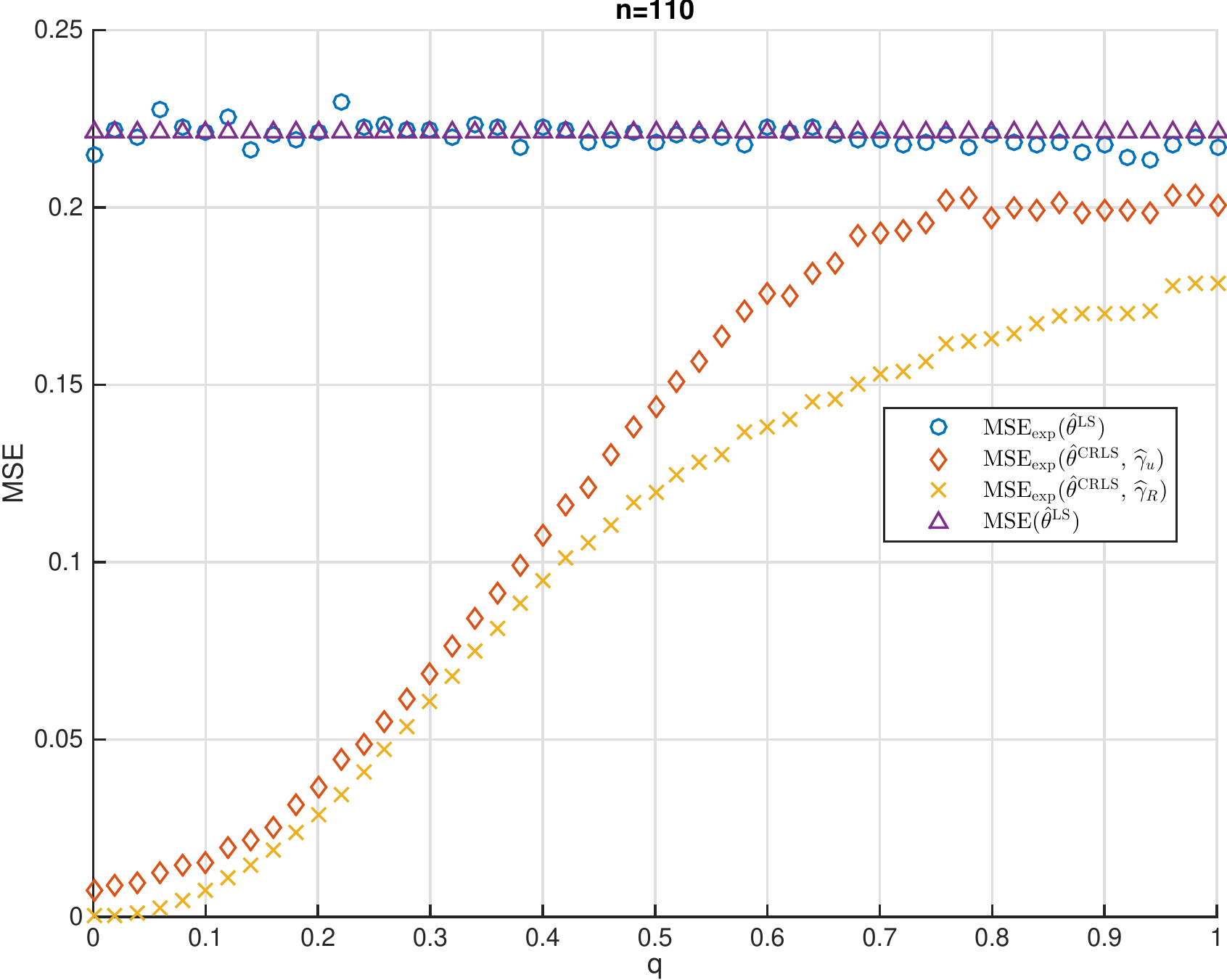}
		\caption{CRLS vs. LS estimates for  Werner states with $n=110$ copies. 
		}\label{fig:r1}
	\end{figure}
	
	\begin{figure}[!htbp]
		\centering
		\includegraphics[width=.7\linewidth]{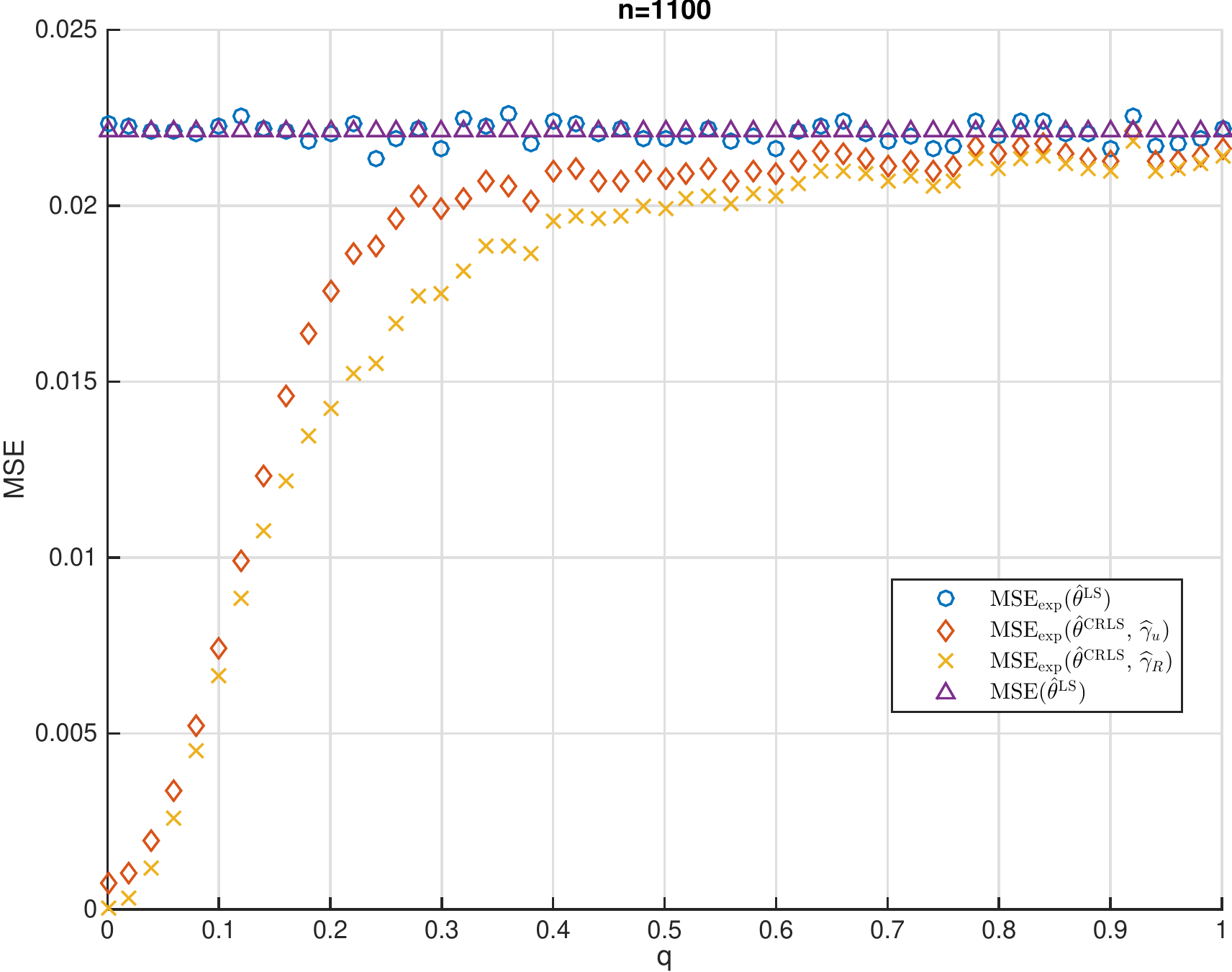}
		\caption{CRLS vs. LS estimates for  Werner states with $n=1100$ copies.}\label{fig:r2}
	\end{figure}

As we can see from the numerical results, with $n=110$, the regularizer for $\widehat{\bm{\theta}}^{\mathrm{CRLS}} $ significantly improves the estimation accuracy compared to $\widehat{\bm{\theta}}^{\mathrm{LS}}$ under both  $\widehat{ \gamma}_{R}$ and $\widehat{ \gamma}_{u}$.  While with $n=1100$, for relatively large $q$, the advantage of $\widehat{\bm{\theta}}^{\mathrm{CRLS}} $ is no longer obvious compared to $\widehat{\bm{\theta}}^{\mathrm{LS}}$ since in this case, the use of the weight ${\rm W}$ becomes essential for the performance. 

\end{exam}

\subsection{Under-determinate Measurement Basis}

\begin{exam}\normalfont
	We consider a 6-qubit quantum state. We use the $6$-qubit Pauli matrices to form our basis $\{\mathsf{B}_j,j=1,\dots,d^2\}$ with $d=2^6=64$.  Let $u = \begin{bmatrix}
	\frac{\sqrt{3}}{2} & \frac{1}{2} \\
	-\frac{1}{2}i & \frac{\sqrt{3}}{2}i
	\end{bmatrix}$ be a $2\times 2$ unitary matrix. Let
	\begin{align*}
	&\bra{\psi_1}=[0,\dots,0,\overbrace{\sqrt{p}}^{i_1\text{-th}},0,\dots,0,\overbrace{\sqrt{1-p}}^{i_2\text{-th}},0,\dots,0]^\top, \\
	& \bra{\psi_2}=[0,\dots,0,\overbrace{1}^{i_3\text{-th}},0,\dots,0]^\top,\\
	& \bra{\psi_3}=[0,\dots,0,\overbrace{1}^{i_4\text{-th}},0,\dots,0]^\top,
	\end{align*} 
	be three pure states, where $i_1=42, i_2=8, i_3=59, i_4=30$. 
	Then $\rho_p=\left(u^{\otimes 6}\right)^\dag(\frac{1}{3}\braket{\psi_1}{\psi_1}+\frac{1}{3}\braket{\psi_2}{\psi_2}+\frac{1}{3}\braket{\psi_3}{\psi_3})u^{\otimes 6}$ is a rank-$3$ density matrix for a $6$-qubit system for all $p\in[0,1]$. Note that $\rho_p$ is low-rank but not sparse due to the existence of $u$.
We index the set of $6$-fold tensor product of Pauli matrices $\big\{\sigma_{l_1}\otimes\dots\otimes\sigma_{l_6}:(l_1,\dots,l_6)\in\{1,2,3\}^6\big\}$   by $\{P_j,j=1,\dots,3^6\}$. The $P_j$'s are of full rank and have eigenvalues $\pm 1$. Denote by $Q_{j\pm}$ the projection onto the eigenspaces of $P_j$ with respect to $\pm 1$ respectively. We randomly choose and then fix $\{Q_{j_1+},\dots,Q_{j_{200}+}\}$ from $\{Q_{j+},j=1,\dots,3^6\}$. Then $\mathcal{M}^+\eq \{M_k=\sqrt{Q_{j_k+}/200},k=1,\dots,200\}$ forms an under-complete  measurement basis, and the resulting $\mathbf{A}=[\beta_1,\dots,\beta_{200}]^\top\in\mathbb{R}^{200\times 4096}$ becomes under-determinate.
	
We use $n=1100,11000,110000$ copies for each $\rho_p$ and perform independent measurements over each copy along any element in the basis $\mathcal{M}^+$. The parameter $p$ is sampled at $p=0,0.1,\dots,1$, and for each $p$, we carry out the tomography procedure for $1000$ rounds. For each round, we set $\gamma=1,10,100,1000$ for $\widehat{\bm{\theta}}^{\mathrm{CRWLS}}$, whose
experimental mean-square errors $\mathbb{MSE}_{\mathrm{exp}}$  with respect to the parametrized $p$   are plotted with $n=1100,11000,110000$ in Figs. \ref{fig:4}--\ref{fig:6}, respectively. From these plots we see that the MSE is fundamentally lower bounded by the $\mathcal{M}^+$ instead of heavily relying on the sample size $n$. Moreover, the MSE and the risk are not quite sensitive with respect to the regularization gain $\gamma$.	It is expected that these estimation results  can be improved by utilizing the optimal regularizer selection from \eqref{unbiasedrisktheta}, but would require significantly higher computation cost. 

	\begin{figure}[!htbp]
		\centering
		\includegraphics[width=.7\linewidth]{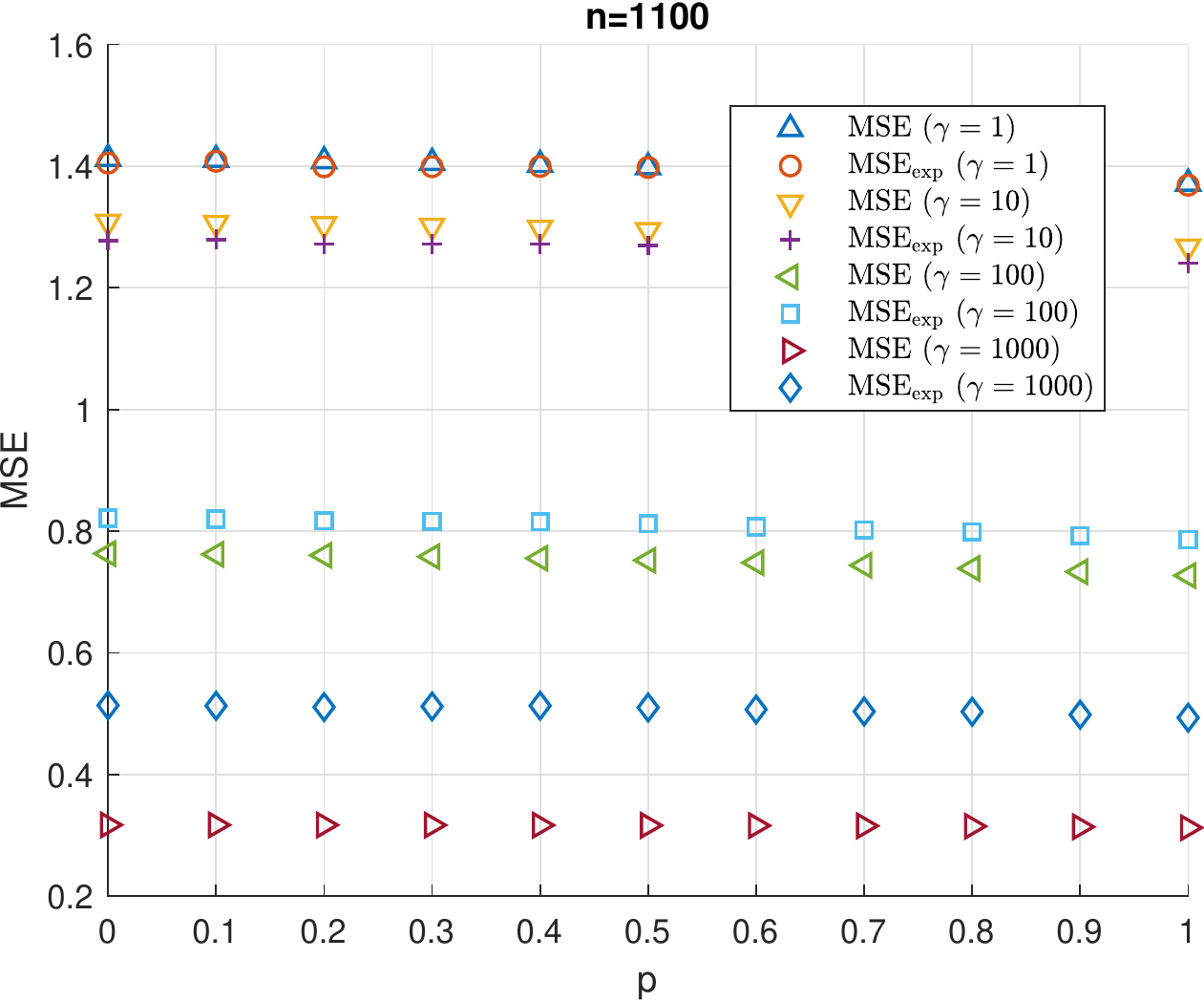}
		\caption{MSEs for estimating the six-qubit state  $\rho_p$ by CRWLS  with $n=1100$ copies.}\label{fig:4}
	\end{figure}
	
	\begin{figure}[!htbp]
		\centering
		\includegraphics[width=.7\linewidth]{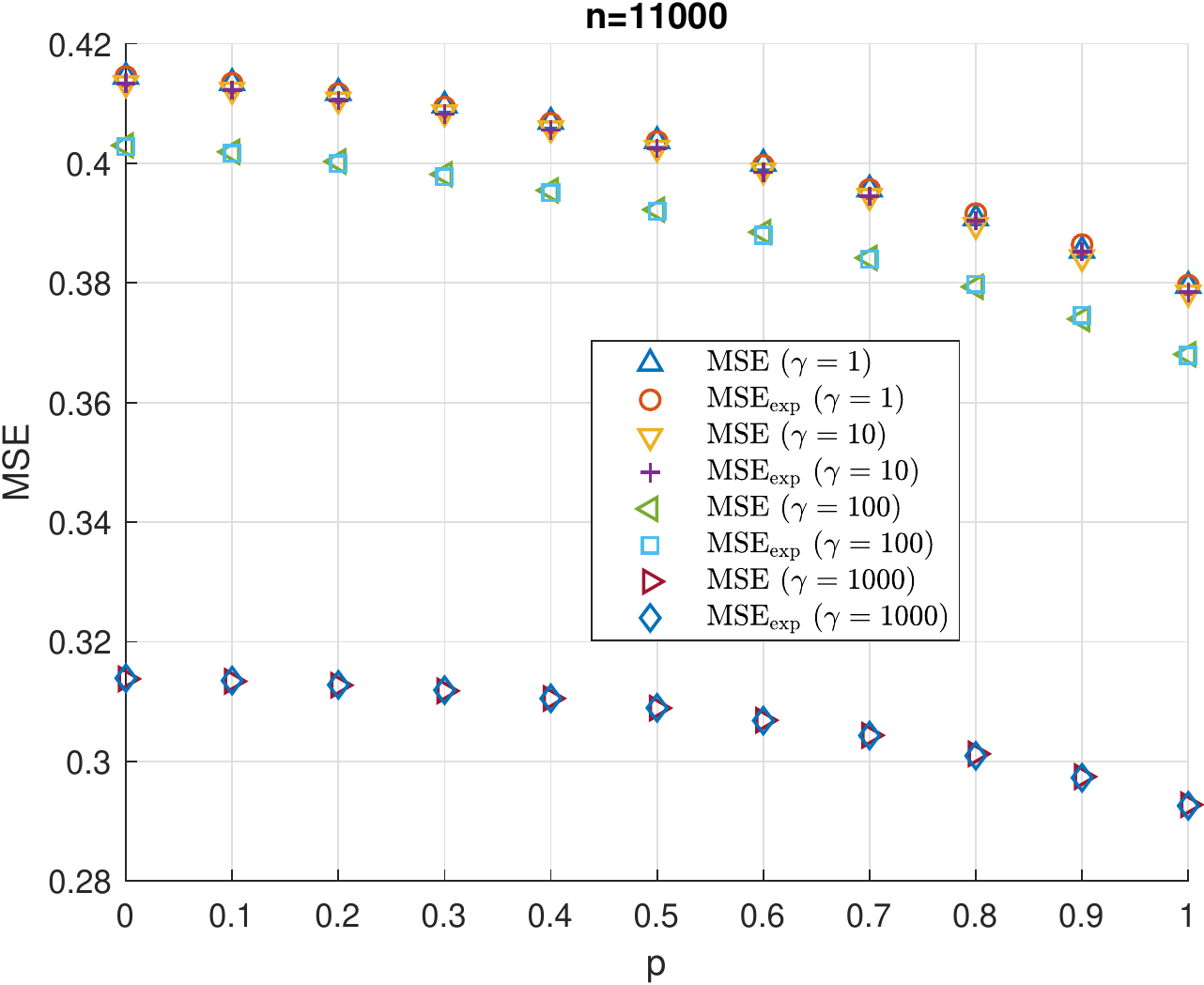}
		\caption{MSEs for estimating the six-qubit state  $\rho_p$ by CRWLS  with $n=11000$ copies.}\label{fig:5}
	\end{figure}
	
	\begin{figure}[!htbp]
		\centering
		\includegraphics[width=.7\linewidth]{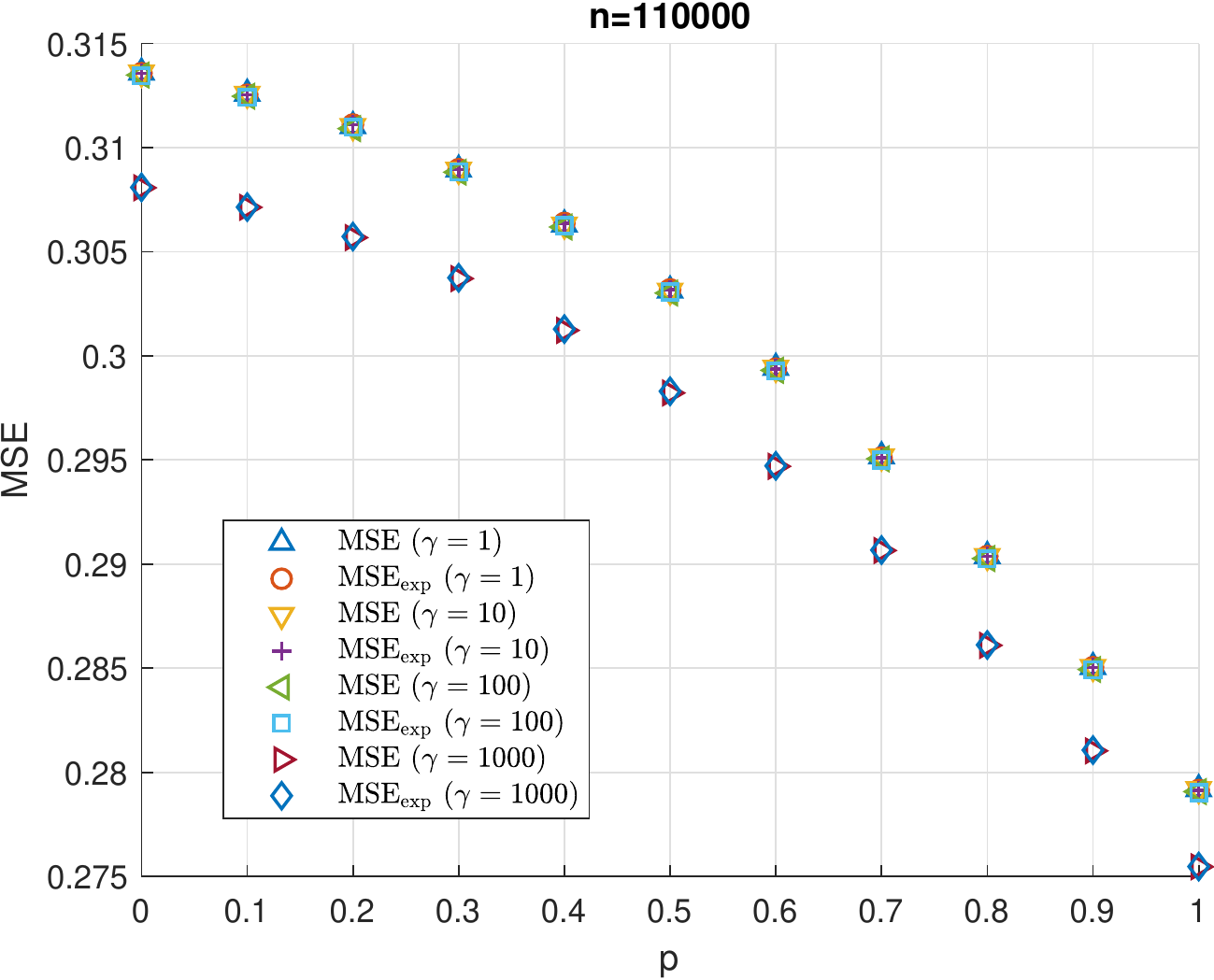}
		\caption{MSEs for estimating the six-qubit state  $\rho_p$ by CRWLS  with $n=110000$ copies.}\label{fig:6}
	\end{figure}
\end{exam}

\section{Conclusions} \label{sec6}

We have studied a series of   linear regression methods for quantum state tomography based on regularization.  With complete or over-complete measurement bases,  the empirical data was shown to be useful  for the construction of a weighted LSE from the measurement outcomes of an unknown quantum state. It was proven that the trace-constrained weighted LSE is the optimal unbiased estimation among all linear estimators. For general measurement bases, either complete or incomplete,  we showed that  $\ell_2$-regularization   with proper regularization parameter could yield even lower mean-square error under a penalty in bias. An explicit formula was established for the regularization parameter under an equivalent regression model, which shows that
the proposed tuning estimators are asymptotically optimal as the number of samples grows to infinity for the risk metric.  An interesting future direction lies in  regularization-based  approach for Hamiltonian identification of quantum dynamical systems.

\medskip

\appendix

\section*{Appendix A}

\renewcommand{\thesection}{A}

\setcounter{thm}{0}
\renewcommand{\thethm}{A\arabic{thm}}
\setcounter{lem}{0}
\renewcommand{\thelem}{A\arabic{lem}}
\setcounter{rem}{0}
\renewcommand{\therem}{A\arabic{rem}}
\renewcommand{\thesection}{A}
\setcounter{equation}{0}
\renewcommand{\theequation}{\thesection.\arabic{equation}}

The proofs of Proposition  \ref{prop3} and Theorem  \ref{thm1} are standard and can be found in the books \citep{Teil1971,Amemiya1985}.
So, they are omitted for saving space.

	\subsection{Proof of Theorem  \ref{thm2}}
The constrained optimization problem \eqref{rop} is transformed to an unconstrained one by introducing Lagrange multiplier $\lambda$
and the resulting Lagrange function of \eqref{rop} is 
\begin{align*}
L(\bm{\theta},\lambda) 
= (\mathbf{y}-\mathbf{A} \bm{\theta})^\top  (\mathbf{y}-\mathbf{A} \bm{\theta})
+ \lambda (\bm{\theta}^\top {\rm Tr}(\mathsf{B})-1).
\end{align*}
Thus the optimal solution $(\widehat{\bm{\theta}}^{\rm CRWLS}, \lambda^*)$ of the problem \eqref{rop} satisfies the first optimality condition
\begin{subequations}
	\begin{align} 
&2 \mathbf{A}^\top \mathrm{W} (\mathbf{y}-\mathbf{A} \widehat{\bm{\theta}}^{\rm CRWLS} ) + 2\gamma \widehat{\bm{\theta}}^{\rm CRWLS}  + \lambda^* \tr(\mathsf{B})=0   \label{crwls1}\\
&\tr(\mathsf{B})^\top \widehat{\bm{\theta}}^{\rm CRWLS}=1. \label{crwls2}
\end{align}
\end{subequations}
It follows from \eqref{crwls1} that 
\begin{align}
\widehat{\bm{\theta}}^{\rm CRWLS} 
=\widehat{\bm{\theta}}^{\rm RWLS}- \lambda^* C\tr(\mathsf{B}).\label{crwls3}
\end{align}
Furthermore, by using \eqref{crwls2}, we have
\begin{align*}
\tr(\mathsf{B})^\top  \widehat{\bm{\theta}}^{\rm CRWLS}= \tr(\mathsf{B})^\top  ( \widehat{\bm{\theta}}^{\rm RWLS}- \lambda^* C \tr(\mathsf{B})) =1,
\end{align*}
which yields
\begin{align}
\lambda^* = \frac{ \tr(\mathsf{B})^\top \widehat{\bm{\theta}}^{\rm CRWLS} -1}{ \tr(\mathsf{B})^\top  C \tr(\mathsf{B})}.\label{oplambda}
\end{align}
Substituting (\ref{oplambda}) into \eqref{crwls3} obtains 
\begin{align}
&\widehat{\bm{\theta}}^{\rm CRWLS}  
= \widehat{\bm{\theta}}^{\rm RWLS} 
-C\nonumber
\tr(\mathsf{B})
\frac{ \tr(\mathsf{B})^\top  \widehat{\bm{\theta}}^{\rm RWLS}  -1}{ \tr(\mathsf{B})^\top  C \tr(\mathsf{B})}.
\end{align}

Next, we compute the MSE matrix of $\widehat{\bm{\theta}}^{\rm CRWLS}$.
By the constraint $\tr(\mathsf{B})^\top \bm{\theta} =1$, we have
\begin{align}
\nonumber
&\hspace{5mm}\widehat{\bm{\theta}}^{\rm CRWLS} 
-\bm{\theta}\\
&= \widehat{\bm{\theta}}^{\rm RWLS} 
-\bm{\theta}
-
\frac{C \tr(\mathsf{B}) }
{ \tr(\mathsf{B})^\top  C \tr(\mathsf{B})}
\big(\tr(\mathsf{B})^\top  (\widehat{\bm{\theta}}^{\rm RWLS}  -\bm{\theta})\big)\nonumber \\
\nonumber
&=C(-\gamma\bm{\theta} + \mathbf{A}^\top\mathrm{W}\mathbf{e})
-
\frac{C \tr(\mathsf{B}) }
{ \tr(\mathsf{B})^\top  C \tr(\mathsf{B})}
\big(\tr(\mathsf{B})^\top
C(-\gamma\bm{\theta} + \mathbf{A}^\top\mathrm{W}\mathbf{e}) \big) \nonumber\\
\nonumber
&=-\gamma F\bm{\theta} + F\mathbf{A}^\top\mathrm{W}\mathbf{e}.
\end{align}
The matrix $F$ has the following properties:
\begin{align}
F^\top = F,~F \tr(\mathsf{B}) = 0,~FC^{-1}F=F. \label{a8}
\end{align}
As a result, the MSE matrix of $\widehat{\bm{\theta}}^{\rm CRWLS} $ is
	\begingroup
\allowdisplaybreaks
\begin{align*}
&\hspace{4mm}\mathbb{E}(\widehat{\bm{\theta}}^{\rm CRWLS}  -\bm{\theta} )
(\widehat{\bm{\theta}}^{\rm CRWLS}  -\bm{\theta} )^\top\\
&=\gamma^2F\bm{\theta}\bm{\theta} ^\top F
+F\mathbf{A}^\top\mathrm{W}\mathbf{A} F\\
&= \gamma^2F\bm{\theta}\bm{\theta} ^\top F
+F(C^{-1}-\gamma \mathsf{I} ) F\\
&=
FC^{-1}F-\gamma F( \mathsf{I}  - \gamma  \bm{\theta}\bm{\theta} ^\top)F\\
&=F-\gamma F(  \mathsf{I}  - \gamma  {\theta}\bm{\theta} ^\top)F.
\end{align*}
\endgroup
This completes the proof.

	\subsection{Proof of Theorem  \ref{thm5}}
	The proof requires the equivalent model \dref{equimodel} as well as resulting notations and conclusion  in Section \ref{sec4}.
Denote the orthogonal matrix appearing in \dref{qmatrix} by
\begin{align}
\mathbf{Q}
=
\left[
\begin{array}{c}
\frac{{\rm Tr}(\mathsf{B})^\top} {\|{\rm Tr}(\mathsf{B})\|}\\
\widetilde{\mathbf{Q}}
\end{array}
\right].
\end{align}	
Thus by \dref{qmatrix} and \dref{equiv} we have
	\begingroup
\allowdisplaybreaks
\begin{align}
\nonumber
\widehat{\bm{\theta}}^{\rm CRWLS}  - \bm{\theta}
&=
\left[\frac{{\rm Tr}(\mathsf{B})} {\|{\rm Tr}(\mathsf{B})\|},~  \widetilde{\mathbf{Q}}^\top \right]
\left[
\begin{array}{c}
\frac{1}{\|{\rm Tr}(\mathsf{B})\|}\\
\widehat{\bm{\alpha}}^{\rm RWLS} 
\end{array}
\right]
-\left[\frac{{\rm Tr}(\mathsf{B})} {\|{\rm Tr}(\mathsf{B})\|},~  \widetilde{\mathbf{Q}}^\top \right]
\left[
\begin{array}{c}
\frac{1}{\|{\rm Tr}(\mathsf{B})\|}\\
\bm{\alpha}
\end{array}
\right]\\
&=\widetilde{\mathbf{Q}}^\top
\big(\widehat{\bm{\alpha}}^{\rm RWLS} 
-   \bm{\alpha}\big).\label{aa2}
\end{align}	
\endgroup
When $0<\gamma<2/\big(\|\bm{\theta}\|^2
-\frac1{\|{\rm Tr}(\mathsf{B})\|^2}\big)$,
 there holds
	\begingroup
\allowdisplaybreaks
\begin{align}
\nonumber
\mathbb{MSE}(&\widehat{\bm{\theta}}^{\rm CRWLS})
\eq
\mathbb{E} (\widehat{\bm{\theta}}^{\rm CRWLS}  - \bm{\theta})
(\widehat{\bm{\theta}}^{\rm CRWLS}  - \bm{\theta})^\top\\
\nonumber
&=
\widetilde{\mathbf{Q}}^\top
\mathbb{E}
\big(\big(\widehat{\bm{\alpha}}^{\rm RWLS} 
- \bm{\alpha}\big)
-\big(\widehat{\bm{\alpha}}^{\rm RWLS} 
-   \bm{\alpha}\big)
\big)^\top
\widetilde{\mathbf{Q}}\\
\nonumber
&=\widetilde{\mathbf{Q}}^\top
	\mathbb{MSE }\big(\widehat{\bm{\alpha}}^{\rm RWLS} (\gamma)\big)
\widetilde{\mathbf{Q}}\\
\nonumber
&<\widetilde{\mathbf{Q}}^\top
\mathbb{MSE }\big(\widehat{\bm{\alpha}}^{\rm RWLS} (0)\big)
\widetilde{\mathbf{Q}}\\
\nonumber
&=\widetilde{\mathbf{Q}}^\top
\mathbb{E}
\big(\big(\widehat{\bm{\alpha}}^{\rm RWLS}(0) 
-\bm{\alpha}\big)
-\big(\widehat{\bm{\alpha}}^{\rm RWLS}(0) 
-   \bm{\alpha}\big)
\big)^\top
\widetilde{\mathbf{Q}}\\
\nonumber
&=
\mathbb{E} (\widehat{\bm{\theta}}^{\rm CRWLS} (0) - \bm{\theta})
(\widehat{\bm{\theta}}^{\rm CRWLS}(0)  - \bm{\theta})^\top\\
&=
\mathbb{MSE }\big(\widehat{\bm{\theta}}^{\rm CWLS}\big)
\end{align}
\endgroup
where the inequality is obtained by \dref{gg1} in Section \ref{sec4} and $\widehat{\bm{\theta}}^{\rm CRWLS} (0)$  is exactly $\widehat{\bm{\theta}}^{\rm CWLS}$.

	\subsection{Proof of Proposition \ref{prop4} and Remark \ref{rem4}}
	It follows from Lemma \ref{lemb1} that 
		\begingroup
	\allowdisplaybreaks
	\begin{align}
	\nonumber
	&\hspace{-5mm}\mathbf{Q} \widehat{\bm{\theta}}^{\rm CRWLS}
	=
	\mathbf{Q} 
	 \mathbf{H} \mathbf{y} + \mathbf{Q}  \mathbf{f}\\
	 \nonumber
	&= \left[
	\begin{array}{c}
	0 \\
	\mathbf{U}
	\end{array}
	\right] 
	\left( \mathbf{z} + \frac{\mathbf{d}} {\|{\rm Tr}(\mathsf{B})\|}\right)
	+\frac{1}{\|{\rm Tr}(\mathsf{B})\|}
	\left[
	\begin{array}{c}
	1 \\
	-\mathbf{U} \mathbf{d}
	\end{array}
	\right]
	\\
	&=
	\left[
	\begin{array}{c}
	\frac1{\|{\rm Tr}(\mathsf{B})\|}\\
	\widehat{\bm{\alpha}}^{\rm RWLS} 
	\end{array}
	\right].
	\label{equi}
	\end{align}
	\endgroup
	Pre-multiplying $\mathbf{Q}^\top$ on both sides of \eqref{equi} proves \eqref{equiv}.

By \eqref{unconstrwls}, we have
\begin{align}
\nonumber
\widehat{\bm{\alpha}}^{\rm RWLS} - \bm{\alpha}
&=\mathbf{V}\mathbf{K} ^\top \mathrm{W} \big(  \mathbf{K} \bm{\alpha} + \mathbf{e} \big)- \bm{\alpha}\\
\nonumber
&=\big(\mathbf{V}\mathbf{K} ^\top \mathrm{W} \mathbf{K} - \mathsf{I} 
\big)\bm{\alpha}
+\mathbf{V}\mathbf{K} ^\top \mathrm{W} \mathbf{e}\\
&=-\gamma\mathbf{V} \bm{\alpha}
+\mathbf{V}\mathbf{K} ^\top \mathrm{W} \mathbf{e}, \label{diff}
\end{align}
which derives
	\begin{align*}
	\mathbb{MSE }\big(\widehat{\bm{\alpha}}^{\rm RWLS} (\gamma)\big)
	&\eq \mathbb{E} (\widehat{\bm{\alpha}}^{\rm RWLS}
	- \bm{\alpha})
	(\widehat{\bm{\alpha}}^{\rm RWLS}
	- \bm{\alpha})^\top\\
	&=\gamma^2  \mathbf{V} \bm{\alpha} \bm{\alpha}^\top \mathbf{V}
	+ \mathbf{V} \mathbf{K}^\top 
	\mathrm{W}
	\mathbf{K} 
	\mathbf{V}.
	\end{align*} 
	When $0<\gamma<2/\bm{\alpha}^\top\bm{\alpha}  $,
	we have $\gamma\bm{\alpha} \bm{\alpha}^\top-2\mathsf{I}<0 $	  and further
	$$\gamma\bm{\alpha} \bm{\alpha}^\top 
	- \gamma (\mathbf{K}^\top 
	\mathrm{W}
	\mathbf{K} )^{-1}
	-2\mathsf{I}<0$$
since $ \gamma (\mathbf{K}^\top 
\mathrm{W}
\mathbf{K} )^{-1}$ is always positive definite.
Thus we obtain
		\begin{align*}
	&\hspace{5mm}\mathbb{MSE }\big(\widehat{\bm{\alpha}}^{\rm RWLS} (\gamma)\big)
	-
	\mathbb{MSE }\big(\widehat{\bm{\alpha}}^{\rm RWLS} (0)\big)\\
	&=\gamma^2  \mathbf{V} \bm{\alpha} \bm{\alpha}^\top \mathbf{V}
	+ \mathbf{V} \mathbf{K}^\top 
	\mathrm{W}
	\mathbf{K} 
	\mathbf{V}
	-
	\big(\mathbf{K}^\top 
	\mathrm{W}
	\mathbf{K} \big)^{-1}\\
	&=\mathbf{V}
	\big(\gamma^2 \bm{\alpha} \bm{\alpha}^\top 
	+  \mathbf{K}^\top 
	\mathrm{W}
	\mathbf{K} 
	-\mathbf{V}^{-1}
	\big(\mathbf{K}^\top 
	\mathrm{W}
	\mathbf{K} \big)^{-1}
	\mathbf{V}^{-1}\big)
	\mathbf{V}\\
	&=\gamma
	\mathbf{V}
	\big(\gamma\bm{\alpha} \bm{\alpha}^\top 
	- \gamma (\mathbf{K}^\top 
	\mathrm{W}
	\mathbf{K} )^{-1}
	-2\mathsf{I}\big)
	\mathbf{V}\\
	&<0
	\end{align*}
	if $0<\gamma<2/\bm{\alpha}^\top  \bm{\alpha}$.

\subsection{Proof of Proposition  \ref{prop5}}
By \eqref{equi}, we have
	\begingroup
\allowdisplaybreaks
\begin{align*}
&\hspace{-9mm}\mathbf{A} \bm{\theta}-   \mathbf{A}\widehat{\bm{\theta}}^{\rm CRWLS} 
 =\mathbf{A}\mathbf{Q}^\top
 \left(\mathbf{Q}  \bm{\theta} 
 -   \mathbf{Q} \widehat{\bm{\theta}}^{\rm CRWLS} \right)\\
 &=\big[\mathbf{d},\mathbf{K}\big]
\left( \left[
 \begin{array}{c}
\frac1 {\|{\rm Tr}(\mathsf{B})\| }\\
\bm{\alpha}
 \end{array}
 \right]   
 -\left[
 \begin{array}{c}
 \frac1{\|{\rm Tr}(\mathsf{B})\|}\\
 \widehat{\bm{\alpha}}^{\rm RWLS} 
 \end{array}
 \right]\right)
 \\
  &= \mathbf{K}\bm{\alpha} - \mathbf{K}\widehat{\bm{\alpha}}^{\rm RWLS} 
\end{align*}
\endgroup
which means that $R(\widehat{\bm{\theta}}^{\rm CRWLS} )=R(\widehat{\bm{\alpha}}^{\rm RWLS} )$.
So, the assertion \eqref{eqhp2} holds.

Similarly, it gives
		\begingroup
\allowdisplaybreaks
\begin{align}
\mathbf{y} -   \mathbf{A}\widehat{\bm{\theta}}^{\rm CRWLS}
=\mathbf{z} + \frac{\mathbf{d}} {\|{\rm Tr}(\mathsf{B})\|}
-
\big[\mathbf{d},\mathbf{K}\big]
\left[
\begin{array}{c}
\frac1{\|{\rm Tr}(\mathsf{B})\|}\\
\widehat{\bm{\alpha}}^{\rm RWLS} 
\end{array}
\right]= \mathbf{z} - \mathbf{K}\widehat{\bm{\alpha}}^{\rm RWLS}.  \label{m1}
\end{align}
\endgroup
The first equation in Lemma \eqref{lemb1} derives
\begin{align}
\tr\big( \mathbf{A} \mathbf{H}  \big)
=
\tr\big( \mathbf{A} \mathbf{Q}^\top \mathbf{Q} \mathbf{H}  \big)
=\tr\left(\big[\mathbf{d},\mathbf{K}\big]
 \left[
\begin{array}{c}
0 \\
\mathbf{U}
\end{array}
\right] \right)
=
\tr\big( \mathbf{K} \mathbf{U}  \big).\label{m2}
\end{align}
Combining \eqref{m1} with \eqref{m2} proves \eqref{eqhp4}.

	\subsection{Proof of Theorem  \ref{thm6}}
	We first prove the convergence \dref{ophp}.
It follows from \eqref{diff} that 
\begin{align*}
\mathbb{E}\big(\mathbf{K} \bm{\alpha} -   \mathbf{K}\widehat{\bm{\alpha}}^{\rm RWLS}  \big)^\top
\mathrm{W} 
  \big(\mathbf{K} \bm{\alpha} -   \mathbf{K}\widehat{\bm{\alpha}}^{\rm RWLS}  \big)
=\gamma^2   \bm{\alpha}^\top \mathbf{V}
\mathbf{K}^\top
\mathrm{W} 
\mathbf{K}
 \mathbf{V} \bm{\alpha}
+ 
\tr\big( 
\mathbf{V} \mathbf{K}^\top 
\mathrm{W}
\mathbf{K} 
\mathbf{V} 
\mathbf{K}^\top
\mathrm{W}
\mathbf{K}  \big)
\end{align*}
 the two terms of which are
\begin{subequations}
	\label{twoterms}
	\begin{align}
\gamma^2   \bm{\alpha}^\top \mathbf{V}
\mathbf{K}^\top
\mathrm{W} 
\mathbf{K}
\mathbf{V} \bm{\alpha}
&= \gamma^2   \bm{\alpha}^\top 
\left(\mathsf{I} - 
\frac{\gamma}{n}
\mathbf{S} \right)
\mathbf{V} \bm{\alpha} 
=\frac{\gamma^2}{n}
\bm{\alpha}^\top
\mathbf{S} \bm{\alpha}
-\frac{\gamma^3}{n^2}
\bm{\alpha}^\top
\mathbf{S}^{2}\bm{\alpha}\label{aa1}\\
%
\mathbf{V} \mathbf{K}^\top 
\mathrm{W}
\mathbf{K} 
\mathbf{V} 
\mathbf{K}^\top
\mathrm{W}
\mathbf{K} & =
\left(  \mathsf{I} - 
\frac{\gamma}{n}
\mathbf{S}  \right)^2
=  \mathsf{I} - 
2\frac{\gamma}{n}
\mathbf{S} 
+\frac{\gamma^2}{n^2}
\mathbf{S}^{2}
\end{align}
\end{subequations}
where $\mathbf{S}$ is defined in \dref{s}.
Define
\begin{align}
\mathscr{C}_R(\gamma)\eq n\left( R(\widehat{\bm{\alpha}}^{\rm RWLS} ) - \tr(\mathsf{I})  \right). \label{lr}
\end{align}
Thus
\begin{align}
\widehat{ \gamma}_{R}
(\widehat{\bm{\alpha}}^{\rm RWLS} )
=\argmin_{\gamma \geq 0}
\mathscr{C}_R(\gamma).
\end{align}
Substituting \eqref{twoterms} into \eqref{lr} yields
\begin{align}
\nonumber
\mathscr{C}_R(\gamma)
&=
\gamma^2
\bm{\alpha}^\top
\mathbf{S}\bm{\alpha}
-\frac{\gamma^3}{n}
\bm{\alpha}^\top
\mathbf{S}^{2}\bm{\alpha}- 
2\gamma
{\rm Tr}\left(\mathbf{S} \right)
+\frac{\gamma^2}{n}
{\rm Tr}\left(\mathbf{S}^{2} \right)\\
&\xra{}
\bm{\alpha} ^\top \Sigma^{-1}\bm{\alpha} \gamma^2
-2  \tr\big(  \Sigma^{-1} \big) \gamma
\eq\mathscr{C}(\gamma)\label{cfcon}
\end{align}
as $n\xra{}\infty$.
It is clear that
	\begin{align*}
	\nonumber
	\gamma^\star
			&\eq
			\argmin_{\gamma\geq 0}
			\mathscr{C}(\gamma)=\frac{\tr\big(  \Sigma^{-1} \big)}
	{\bm{\alpha} ^\top \Sigma^{-1}\bm{\alpha} }
	\end{align*}
	which can also be expressed by $\bm{\theta}$ and $\Upsilon$ in terms of Lemma \ref{lemb3}.
	
By Lemma  \ref{ct}, the limit $\widehat{ \gamma}_{R}
(\widehat{\bm{\alpha}}^{\rm RWLS} )\xra{}\gamma^\star$ holds
as $n\xra{}\infty$ since the convergence \dref{cfcon} is uniform over
a compact subset of $[0,+\infty)$ that includes 
 $\gamma^\star$.

For convenience of proving \eqref{ophpr}, denote
\begin{align}
\widehat{\bm{\alpha}}^{\rm WLS} 
\eq 
\mathbf{X}
\mathbf{K}^\top  \mathrm{W}
\mathbf{z},~~
\mathbf{X}\eq(\mathbf{K}^\top  \mathrm{W}\mathbf{K})^{-1}.
\end{align}
It follows that
	\begingroup
\allowdisplaybreaks
\begin{align}
\mathbf{V}
&=\mathbf{X}
-\gamma \mathbf{V} \mathbf{X},~~
\mathbf{V}\mathbf{X} ^{-1}
=\mathsf{I}
-\gamma \mathbf{V} \\
\nonumber
\mathbf{K}\widehat{\bm{\alpha}}^{\rm RWLS}
&=\mathbf{K} \mathbf{V}\mathbf{K}^\top
\mathrm{W} \mathbf{z}
=\mathbf{K}
\big(\mathbf{X}
-\gamma \mathbf{V} \mathbf{X}\big)\mathbf{K}^\top
\mathrm{W} \mathbf{z}\\
&=\mathbf{K} \widehat{\bm{\alpha}}^{\rm WLS}
-\gamma\mathbf{K} \mathbf{V}\widehat{\bm{\alpha}}^{\rm WLS}.
\end{align}
\endgroup
We first consider the decomposition of the first term of the cost function of  \eqref{unbiasedriskalpha}
\begin{align}
\nonumber
&\hspace{5mm}(\mathbf{z} - \mathbf{K}\widehat{\bm{\alpha}}^{\rm RWLS})^\top
\mathrm{W}
(\mathbf{z} - \mathbf{K}\widehat{\bm{\alpha}}^{\rm RWLS})\\
\nonumber
&
=\big(\mathbf{z} -  \mathbf{K} \widehat{\bm{\alpha}}^{\rm WLS}\big)^\top 
\mathrm{W}
\big(\mathbf{z} -  \mathbf{K} \widehat{\bm{\alpha}}^{\rm WLS}\big)
+\gamma^2(\widehat{\bm{\alpha}}^{\rm WLS})^\top
\mathbf{V} \mathbf{K}^\top
\mathrm{W}
\mathbf{K}
\mathbf{V}
\widehat{\bm{\alpha}}^{\rm WLS}\\
&\hspace{5mm}+2\gamma 
(\widehat{\bm{\alpha}}^{\rm WLS})^\top
\mathbf{V}
\mathbf{K}^\top
\mathrm{W}
\big(\mathbf{z} -  \mathbf{K} \widehat{\bm{\alpha}}^{\rm WLS} \big).\label{ft}
\end{align}
The second term of \eqref{ft} is
\begin{align}\nonumber
&~~~~\gamma^2(\widehat{\bm{\alpha}}^{\rm WLS})^\top
\mathbf{V} \mathbf{K}^\top
\mathrm{W}
\mathbf{K}
\mathbf{V}
\widehat{\bm{\alpha}}^{\rm WLS}\\
\nonumber
&=
\gamma^2
(\widehat{\bm{\alpha}}^{\rm WLS})^\top
\big(\mathsf{I} - \gamma \mathbf{V} \big) \mathbf{V}
\widehat{\bm{\alpha}}^{\rm WLS}\\
\nonumber
&=
\gamma^2
(\widehat{\bm{\alpha}}^{\rm WLS})^\top
\mathbf{V}
(\widehat{\bm{\alpha}}^{\rm WLS})
-\gamma^3
(\widehat{\bm{\alpha}}^{\rm WLS})^\top
\mathbf{V}^2
\widehat{\bm{\alpha}}^{\rm WLS}\\
&=
\frac{\gamma^2}{n}
(\widehat{\bm{\alpha}}^{\rm WLS})^\top
\mathbf{S}
(\widehat{\bm{\alpha}}^{\rm WLS})
-\frac{\gamma^3}{n^2}
(\widehat{\bm{\alpha}}^{\rm WLS})^\top
\mathbf{S} ^{2}
\widehat{\bm{\alpha}}^{\rm WLS}.
\label{c1}
\end{align}
The third term of \eqref{ft} is
\begin{align}
2\gamma 
(\widehat{\bm{\alpha}}^{\rm WLS})^\top
\mathbf{V}
\mathbf{K}^\top
\mathrm{W}
\big(\mathbf{z} -  \mathbf{K} \widehat{\bm{\alpha}}^{\rm WLS} \big)
=2\gamma 
(\widehat{\bm{\alpha}}^{\rm WLS})^\top
\mathbf{V}
\big( \mathbf{K}^\top
\mathrm{W} \mathbf{z}
-
\mathbf{K}^\top
\mathrm{W} 
\mathbf{K} \widehat{\bm{\alpha}}^{\rm WLS}
\big)
=0.
\label{c2}
\end{align}
Further, we have
\begin{align}
\mathbf{K} \mathbf{U} 
=
\mathbf{V} \mathbf{K}^\top 
\mathrm{W}
\mathbf{K} 
=
 \mathsf{I} - 
\frac{\gamma}{n}
\mathbf{S} .
\label{c3}
\end{align}
Define
\begin{align*}
\nonumber
\mathscr{C}_U(\gamma) \eq n\big( 
\mathscr{C}_u(\gamma)
&-\big(\mathbf{z} \!-\! \mathbf{K} \widehat{\bm{\alpha}}^{\rm WLS}\big)^\top 
\mathrm{W}
\big(\mathbf{z} \!-\!  \mathbf{K} \widehat{\bm{\alpha}}^{\rm WLS}\big)
-2\tr(\mathsf{I}) \big). \label{go}
\end{align*}
Thus
\begin{align}
\widehat{ \gamma}_{u}
(\widehat{\bm{\alpha}}^{\rm RWLS} )
=\argmin_{\gamma \geq 0}
\mathscr{C}_U(\gamma).
\end{align}
since $\big(\mathbf{z} -  \mathbf{K} \widehat{\bm{\alpha}}^{\rm WLS}\big)^\top 
\mathrm{W}
\big(\mathbf{z} -  \mathbf{K} \widehat{\bm{\alpha}}^{\rm WLS}\big)$ is independent of $\gamma$.
Substituting \eqref{c1}--\eqref{c3} into \eqref{go} turns out
	\begingroup
\allowdisplaybreaks
\begin{align}
\nonumber
\mathscr{C}_U(\gamma)
&=
\gamma^2
(\widehat{\bm{\alpha}}^{\rm WLS})^\top
\mathbf{S}
\widehat{\bm{\alpha}}^{\rm WLS}
-\frac{\gamma^3}{n}
(\widehat{\bm{\alpha}}^{\rm WLS})^\top
\mathbf{S} ^{2}
\widehat{\bm{\alpha}}^{\rm WLS}- 
2\gamma
{\rm Tr}\left(\mathbf{S} \right)\\
&\xra{}
\bm{\alpha} ^\top \Sigma^{-1}\bm{\alpha} \gamma^2
-2  \tr\big(  \Sigma^{-1} \big) \gamma
=\mathscr{C}(\gamma)
\end{align}
\endgroup
as $n\xra{}\infty$ since $\widehat{\bm{\alpha}}^{\rm WLS} \xra{} \bm{\alpha}$ almost surely as $n\xra{}\infty$.
It follows from Lemma \ref{ct} that the limit \eqref{ophpr} is true.

It remains to show the rates of convergence \dref{ophpd} and \dref{ophpr}.

For this, we first calculate the first and second order derivatives of $\mathscr{C}_R(\gamma)$ with respective to $\gamma$ with the help of Lemma \ref{lemb2}:
\begin{subequations}
	\begin{align*}
\nonumber
\frac{{\rm d}~\!\mathscr{C}_R(\gamma)}{{\rm d}~\! \gamma}
=~\!&2\gamma
\bm{\alpha}^\top
\mathbf{S}
\bm{\alpha}
-\frac{\gamma^2}{n}
\bm{\alpha}^\top
\mathbf{S^2}
\bm{\alpha}
-\frac{3\gamma^2 }{n}
\bm{\alpha}^\top
\mathbf{S}^2
\bm{\alpha}
+\frac{2\gamma^3}{n^2}
\bm{\alpha}^\top
\mathbf{S^3}
\bm{\alpha}\\
\nonumber
&-
2\tr\big( \mathbf{S} \big)
+\frac{2\gamma }{n}
\tr\big( \mathbf{S}^2 \big)+\frac{2\gamma }{n}
\tr\big( \mathbf{S}^2 \big)
-\frac{2\gamma^2 }{n^2}
\tr\big( \mathbf{S}^3 \big)\\
\frac{{\rm d^2}~\!\mathscr{C}_R(\gamma)}{{\rm d}~\! \gamma^2}
=~\!&2
\bm{\alpha}^\top
\mathbf{S}
\bm{\alpha}
+O\left(\frac{1}n\right).
\end{align*}
\end{subequations}
By using a Taylor expansion, we have	
\begin{align*}
0=\frac{{\rm d}~\!\mathscr{C}_R(\gamma)}{{\rm d}~\! \gamma}
\Big|_{\gamma = \widehat{ \gamma}_{R}
	(\widehat{\bm{\alpha}}^{\rm RWLS} )}
=\frac{{\rm d}~\!\mathscr{C}_R(\gamma)}{{\rm d}~\! \gamma}
\Big|_{\gamma = \gamma^\star}
+
\frac{{\rm d^2}~\!\mathscr{C}_R(\gamma)}{{\rm d}~\! \gamma^2}
\Big|_{\gamma = \overline{\gamma}  }
\left(\widehat{ \gamma}_{R}
(\widehat{\bm{\alpha}}^{\rm RWLS} )
-   \gamma^\star\right).
\end{align*}	
where $\overline{\gamma}$ is a real number between $\widehat{ \gamma}_{R}
(\widehat{\bm{\alpha}}^{\rm RWLS} )$ and $\gamma^\star$,
which implies that
\begin{align*}
\widehat{ \gamma}_{R}
(\widehat{\bm{\alpha}}^{\rm RWLS} )
-   \gamma^\star
=-
\left(
\frac{{\rm d^2}~\!\mathscr{C}_R(\gamma)}{{\rm d}~\! \gamma^2}
\Big|_{\gamma = \overline{\gamma}  }
\right)^{-1}
\frac{{\rm d}~\!\mathscr{C}_R(\gamma)}{{\rm d}~\! \gamma}
\Big|_{\gamma = \gamma^\star}.
\end{align*}
As  $n\xra{}\infty$, we have
	\begingroup
\allowdisplaybreaks
\begin{subequations}
	\begin{align*}
\nonumber
n\frac{{\rm d}~\!\mathscr{C}_R(\gamma)}{{\rm d}~\! \gamma}
\Big|_{\gamma = \gamma^\star}
&=	2n\big( \gamma^\star
\bm{\alpha}^\top
\mathbf{S}
\bm{\alpha} - \tr\big( \mathbf{S} \big)  \big)
-\frac{4(\gamma^\star)^2}{n}
\bm{\alpha}^\top
\mathbf{S^2}
\bm{\alpha}
+\frac{4\gamma^\star }{n}
\tr\big( \mathbf{S}^2 \big)
+O\left(\frac1n\right)\\
&
\nonumber
=-2(\gamma^\star)^2
\bm{\alpha}^\top
\mathbf{S}
\Sigma^{-1}
\bm{\alpha}
+2\gamma^\star
\tr\big( \mathbf{S} \Sigma^{-1} \big)
-4(\gamma^\star)^2
\bm{\alpha}^\top
\mathbf{S^2}
\bm{\alpha}
+4\gamma^\star
\tr\big( \mathbf{S}^2 \big)
+O\left(\frac1n\right)\\
&\xra{}
-6(\gamma^\star)^2
\bm{\alpha}^\top
\Sigma^{-2}
\bm{\alpha}
+6\gamma^\star
\tr\big( \Sigma^{-2} \big)\\
&\hspace{-3mm}\frac{{\rm d^2}~\!\mathscr{C}_R(\gamma)}{{\rm d}~\! \gamma^2}
\Big|_{\gamma = \overline{\gamma}  }
\xra{}
2
\bm{\alpha}^\top
\Sigma^{-1}
\bm{\alpha}
\end{align*}	
\end{subequations}
\endgroup
which yields
	\begingroup
\allowdisplaybreaks
\begin{align}
n\big(\widehat{ \gamma}_{R}
(\widehat{\bm{\alpha}}^{\rm RWLS} )
-   \gamma^\star\big)\xra{}
\frac{3\gamma^\star
\big(\gamma^\star
\bm{\alpha}^\top
\Sigma^{-2}
\bm{\alpha}
-
\tr\big( \Sigma^{-2} \big)
\big)}
{\bm{\alpha}^\top
	\Sigma^{-1}
	\bm{\alpha}}.
\end{align}
\endgroup
For proving \dref{ophpr}, the procedure is similar.
By Lemma \ref{lemb2}, the first and second derivatives of $\mathscr{C}_U(\gamma)$ are
\begingroup
\allowdisplaybreaks
\begin{align*}
\nonumber
\frac{{\rm d}~\!\mathscr{C}_U(\gamma)}{{\rm d}~\! \gamma}
&=~\!2\gamma
(\widehat{\bm{\alpha}}^{\rm WLS})^\top
\mathbf{S}
\widehat{\bm{\alpha}}^{\rm WLS}
-\frac{\gamma^2}{n}
(\widehat{\bm{\alpha}}^{\rm WLS})^\top
\mathbf{S^2}
\widehat{\bm{\alpha}}^{\rm WLS}\\
\nonumber
&\hspace{-3mm}
-\frac{3\gamma^2 }{n}
(\widehat{\bm{\alpha}}^{\rm WLS})^\top
\mathbf{S}^2
\widehat{\bm{\alpha}}^{\rm WLS}
+\frac{2\gamma^3}{n^2}
(\widehat{\bm{\alpha}}^{\rm WLS})^\top
\mathbf{S^3}
\widehat{\bm{\alpha}}^{\rm WLS}\\
&\hspace{-3mm}
-
2\tr\big( \mathbf{S} \big)
+\frac{2\gamma }{n}
\tr\big( \mathbf{S}^2 \big)\\
\frac{{\rm d^2}~\!\mathscr{C}_U(\gamma)}{{\rm d}~\! \gamma^2}
&=2
(\widehat{\bm{\alpha}}^{\rm WLS})^\top
\mathbf{S}
\widehat{\bm{\alpha}}^{\rm WLS}
+O_p\left(\frac{1}n\right).
\end{align*}	
\endgroup
Applying the Taylor expansion obtains
\begin{align*}
\widehat{ \gamma}_{u}
(\widehat{\bm{\alpha}}^{\rm RWLS} )
-   \gamma^\star
=-
\left(
\frac{{\rm d^2}~\!\mathscr{C}_U(\gamma)}{{\rm d}~\! \gamma^2}
\Big|_{\gamma = \overline{\gamma}  }
\right)^{-1}
\frac{{\rm d}~\!\mathscr{C}_U(\gamma)}{{\rm d}~\! \gamma}
\Big|_{\gamma = \gamma^\star}
\end{align*}
where $\overline{\gamma}$ is a real number between $\widehat{ \gamma}_{u}
(\widehat{\bm{\alpha}}^{\rm RWLS} )$ and $\gamma^\star$.
By a straightforward calculation,  we have
\begin{align*}
\sqrt{n}\big(\widehat{\bm{\alpha}}^{\rm WLS}
-  
\bm{\alpha} \big)
\xra{}\mathscr{N}
\left( 0, \Sigma^{-1} \right)
\end{align*}
by further using the Delta method, 
\begin{align*}
\nonumber
\sqrt{n}\big( (\widehat{\bm{\alpha}}&^{\rm WLS})^\top
\mathbf{S}
\widehat{\bm{\alpha}}^{\rm WLS}
-  \bm{\alpha}^\top
\Sigma^{-1}
\bm{\alpha} \big)\\
\nonumber
=~\!&
\sqrt{n}
\big(\widehat{\bm{\alpha}}^{\rm WLS}
-
\bm{\alpha}\big)^\top
\mathbf{S}
\widehat{\bm{\alpha}}^{\rm WLS}+\sqrt{n}
\bm{\alpha}^\top
\big( \mathbf{S} - \Sigma^{-1}  \big)
\widehat{\bm{\alpha}}^{\rm WLS}
+\bm{\alpha}^\top
\Sigma^{-1}
\sqrt{n}\big(\widehat{\bm{\alpha}}^{\rm WLS}
-  
\bm{\alpha} \big)\\
\nonumber
=~\!&
\sqrt{n}
\big(\widehat{\bm{\alpha}}^{\rm WLS}
-
\bm{\alpha}\big)^\top
\mathbf{S}
\widehat{\bm{\alpha}}^{\rm WLS}
+\bm{\alpha}^\top
\Sigma^{-1}
\sqrt{n}\big(\widehat{\bm{\alpha}}^{\rm WLS}
-  
\bm{\alpha} \big)
+O_p\left(\frac{1}{\sqrt{n}} \right)\\
\xra{}~\!&
\mathscr{N}
\left( 0, 4~\!\bm{\alpha}^\top
\Sigma^{-3} \bm{\alpha}\right).
\end{align*}
It follows that
\begin{align*}
\nonumber
&\hspace{5mm}\sqrt{n}
\frac{{\rm d}~\!\mathscr{C}_U(\gamma)}{{\rm d}~\! \gamma}
\Big|_{\gamma = \gamma^\star}\\
\nonumber
&=2\gamma^\star 
(\widehat{\bm{\alpha}}^{\rm WLS})^\top
\mathbf{S}
\widehat{\bm{\alpha}}^{\rm WLS}
-
2\tr\big( \mathbf{S} \big)
+O_p\left(\frac{1}{\sqrt{n}} \right)\\
\nonumber
&=2 \sqrt{n}
\gamma^\star 
\left(  (\widehat{\bm{\alpha}}^{\rm WLS})^\top
\mathbf{S}
\widehat{\bm{\alpha}}^{\rm WLS}
-  \bm{\alpha}^\top
\Sigma^{-1}
\bm{\alpha} \right)
+ 2\sqrt{n}\big(\tr\big( \Sigma^{-1} \big) 
-2\tr\big( \mathbf{S} \big)\big)
+O_p\left(\frac{1}{\sqrt{n}} \right)\\
&=2 
\gamma^\star 
\sqrt{n}
\left(  (\widehat{\bm{\alpha}}^{\rm WLS})^\top
\mathbf{S}
\widehat{\bm{\alpha}}^{\rm WLS}
-  \bm{\alpha}^\top
\Sigma^{-1}
\bm{\alpha} \right)
+O_p\left(\frac{1}{\sqrt{n}} \right)\\
&\xra{}\mathscr{N}
\left(0, 16(\gamma^\star)^2\bm{\alpha}^\top
	\Sigma^{-3}
	\bm{\alpha} 
\right).
\end{align*}
As a result,
\begin{align*}
&\sqrt{n}
\big(	\widehat{ \gamma}_u 
(\widehat{\bm{\alpha}}^{\rm RWLS} ) 
-\gamma^\star\big)
\xra{}
\mathscr{N}
\left(0, \frac{4(\gamma^\star)^2\bm{\alpha}^\top
	\Sigma^{-3}
	\bm{\alpha} }
{\big(\bm{\alpha}^\top
	\Sigma^{-1}
	\bm{\alpha}\big)^2 }
 \right).
\end{align*}

	\section*{Appendix B}
	
	\renewcommand{\thesection}{B}
	
	\setcounter{equation}{0}
	\renewcommand{\theequation}{\thesection.\arabic{equation}}
	\setcounter{lem}{0}
	\renewcommand{\thelem}{B\arabic{lem}}
	\setcounter{rem}{0}
	\renewcommand{\therem}{B\arabic{rem}}

	This appendix contains the technical lemmas used in the proof in
	Appendix A.
	\setcounter{subsection}{0}
\begin{lem} 
	\label{lemb1}
We have
	\begin{align}
	\mathbf{Q} \mathbf{H}
	=
	\left[
	\begin{array}{c}
	0 \\
	\mathbf{U}
	\end{array}
	\right],~~
	\mathbf{Q} \bm{ f}
	=
	\frac{1}{\|{\rm Tr}(\mathsf{B})\|}
	\left[
	\begin{array}{c}
	1 \\
	-\mathbf{U} \mathbf{d}
	\end{array}
	\right].\label{id4}
	\end{align}	
	
\end{lem}

\begin{proof}
	For convenience of proof, we denote
	\begin{align*}
	&C=(\mathbf{A}^\top  \mathrm{W}\mathbf{A} + \gamma \mathsf{I})^{-1}\\
	&\mathbf{G} = (\mathbf{D}^\top  \mathrm{W}\mathbf{D} + \gamma \mathsf{I})^{-1}
	\end{align*}
	as well as the column vector $\bm{e_1}$  of length $d^2-1$ and matrix $\bm{e_{-1}}$ of size $d^2$ by $d^2-1$  that make up the identity matrix
	\begingroup
	\allowdisplaybreaks
	\begin{align*}
	\left[
	\begin{array}{cc}
	\bm{e_1}, &\bm{e_{-1}} 
	\end{array}
	\right] 
	=\mathsf{I}.
	\end{align*}
	\endgroup
	Thus the following identities hold.
	\begingroup
	\allowdisplaybreaks
	\begin{subequations}
		\begin{align}
		&	\mathbf{G} =\mathbf{Q} C \mathbf{Q}^\top,~
		\mathbf{D} = \mathbf{A}\mathbf{Q}^\top=[\mathbf{d},\mathbf{K}] \label{id1}\\
		&\bm{e_1}^\top
		\mathbf{Q} 
		={\rm Tr}(\mathsf{B})^\top\!/\|{\rm Tr}(\mathsf{B})\|,~
		\mathbf{Q} {\rm Tr}(\mathsf{B})
		\!=\!\|{\rm Tr}(\mathsf{B})\|\bm{e_1}\label{id2}\\
		&\bm{e_{-1}}^\top\mathbf{G}\bm{e_1}/\bm{e_1}^\top \mathbf{G}\bm{e_1}
		\!=\!-(\mathbf{K}^\top  \mathrm{W}\mathbf{K} + \gamma \mathsf{I})^{-1}\mathbf{K} ^\top \mathrm{W} \mathbf{d}. \label{id3}
		\end{align}	
	\end{subequations}
	\endgroup
	
	The identities \eqref{id1}--\eqref{id2} can be verified by a straightforward calculation.
	
	The  equation \eqref{id3} is obtained by choosing the $(2,1)$-block submatrix of the identity 
	\begingroup
	\allowdisplaybreaks
	\begin{align*}
	\mathbf{G}^{-1}\mathbf{G}
	=\left[
	\begin{array}{cc} 
	\mathbf{d}^\top \mathrm{W} \mathbf{d} +\gamma & \mathbf{d}^\top \mathrm{W} \mathbf{K} \\
	\mathbf{K}^\top \mathrm{W} \mathbf{d} & \mathbf{K}^\top \mathrm{W} \mathbf{K} +\gamma\mathsf{I}
	\end{array}
	\right]\times \left[
	\begin{array}{cc}
	\bm{e_1}^\top \mathbf{G}\bm{e_1} & \bm{e_1}^\top\mathbf{G}\bm{e_{-1}}\\
	\bm{e_{-1}}^\top\mathbf{G}\bm{e_1}& \bm{e_{-1}}^\top\mathbf{G}\bm{e_{-1}}
	\end{array}
	\right]
	=\mathsf{I}.
	\end{align*}
	\endgroup
    By using \eqref{id1}--\eqref{id2}, we have
    \begingroup
    \allowdisplaybreaks
	\begin{align}
	\nonumber
	&\hspace{4.5mm}\mathbf{Q} \mathbf{H}
	=\mathbf{Q}C\mathbf{A}^\top  \mathrm{W} 
	-\mathbf{Q}C \tr(\mathsf{B}) \nonumber
	\frac{ \tr(\mathsf{B})^\top  C \mathbf{A}^\top  \mathrm{W} }{ \tr(\mathsf{B})^\top  C\tr(\mathsf{B})}\\
	\nonumber
	& =\mathbf{Q}C\mathbf{Q}^\top \mathbf{Q}\mathbf{A}^\top  \mathrm{W}
	-\mathbf{Q}C \mathbf{Q}^\top 
	\mathbf{Q} \tr(\mathsf{B}) \nonumber
	\frac{ \tr(\mathsf{B})^\top \mathbf{Q}^\top \mathbf{Q} C 
		\mathbf{Q}^\top \mathbf{Q}\mathbf{A}^\top  \mathrm{W} }
	{ \tr(\mathsf{B})^\top \mathbf{Q}^\top \mathbf{Q}  C \mathbf{Q}^\top \mathbf{Q}\tr(\mathsf{B})}\\
	\nonumber
	&=\mathbf{G} \mathbf{D}\mathrm{W} 
	-\mathbf{G}\bm{e_1} \bm{e_1}^\top\mathbf{G} \mathbf{D}\mathrm{W} /\bm{e_1}^\top \mathbf{G}\bm{e_1}.
	\end{align}
	\endgroup
It is clear that
	\begin{align}
\nonumber
&\hspace{4.5mm}
\bm{e_{1}}^\top\mathbf{Q} \mathbf{H}
=\bm{e_{1}}^\top\mathbf{G} \mathbf{D}\mathrm{W} 
-\bm{e_{1}}^\top\mathbf{G}\bm{e_1} \bm{e_1}^\top\mathbf{G} \mathbf{D}\mathrm{W} /\bm{e_1}^\top \mathbf{G}\bm{e_1}=0.
\end{align}
and further it follows from \eqref{id3} that
\begingroup
\allowdisplaybreaks
		\begin{align*}
	\nonumber
	&\hspace{4.5mm}
	\bm{e_{-1}}^\top\mathbf{Q} \mathbf{H}
	=\bm{e_{-1}}^\top\mathbf{G} \mathbf{D}\mathrm{W} 
	-\bm{e_{-1}}^\top\mathbf{G}\bm{e_1} \bm{e_1}^\top\mathbf{G} \mathbf{D}\mathrm{W} /\bm{e_1}^\top \mathbf{G}\bm{e_1}\\
	&=\left( \bm{e_{-1}}^\top 
	-\frac{\bm{e_{-1}}^\top\mathbf{G}\bm{e_1}}{\bm{e_1}^\top \mathbf{G}\bm{e_1}} 
	\bm{e_1}^\top
	\right)
	\mathbf{G} \mathbf{D}\mathrm{W}\\
	&=
	\left[
	\begin{array}{cc}
	-\frac{\bm{e_{-1}}^\top\mathbf{G}\bm{e_1}}{\bm{e_1}^\top \mathbf{G}\bm{e_1}}, & \mathsf{I}
	\end{array}
	\right]
	\left[
	\begin{array}{c}
	\bm{e_1}^\top\\
	\bm{e_{-1}}^\top
	\end{array}
	\right] 
	\mathbf{G}
	\mathbf{D}\mathrm{W} \\
	&=
	\left[
	\begin{array}{cc}
	(\mathbf{K}^\top  \mathrm{W}\mathbf{K} + \gamma \mathsf{I})^{-1}\mathbf{K} ^\top \mathrm{W} \mathbf{d}, & \mathsf{I}
	\end{array}
	\right]
	\mathbf{G}
	\mathbf{D}^\top\mathrm{W} \\
	&=(\mathbf{K}^\top  \mathrm{W}\mathbf{K} + \gamma \mathsf{I})^{-1}
	\left[
	\begin{array}{cc}
	\mathbf{K} ^\top \mathrm{W} \mathbf{d},
	& \mathbf{K}^\top  \mathrm{W}\mathbf{K} + \gamma \mathsf{I}
	\end{array}
	\right]
	\mathbf{G}
	\mathbf{D}^\top \mathrm{W} \\
	& =(\mathbf{K}^\top  \mathrm{W}\mathbf{K} + \gamma \mathsf{I})^{-1}
	\bm{e_{-1}}^\top  \mathbf{G}^{-1} \mathbf{G}\mathbf{D}^\top\mathrm{W} \\
	& =(\mathbf{K}^\top  \mathrm{W}\mathbf{K} + \gamma \mathsf{I})^{-1}
	\bm{e_{-1}}^\top  \mathbf{D}^\top\mathrm{W} \\
	&=(\mathbf{K}^\top  \mathrm{W}\mathbf{K} + \gamma \mathsf{I})^{-1}
	\mathbf{K}^\top\mathrm{W}\\
	&=\mathbf{U}.
	\end{align*}
	\endgroup
Using \eqref{id1}--\eqref{id2} again, one yields
\begingroup
\allowdisplaybreaks
\begin{align*}
\mathbf{Q} \bm{ f}
&=
 \frac{\mathbf{Q}  C \tr(\mathsf{B})}
{ \tr(\mathsf{B})^\top  C \tr(\mathsf{B})}\\
&=
\frac{\mathbf{Q}  C \mathbf{Q}^\top \mathbf{Q} \tr(\mathsf{B})}
{ \tr(\mathsf{B})^\top \mathbf{Q}^\top \mathbf{Q} C \mathbf{Q}^\top \mathbf{Q} \tr(\mathsf{B})}\\
&=
\frac{\mathbf{G}\bm{e_1}}{\bm{e_1}^\top \mathbf{G}\bm{e_1}\|{\rm Tr}(\mathsf{B})\|}.
\end{align*}
	\endgroup
	Thus by \eqref{id3} we have
		\begingroup
	\allowdisplaybreaks
	\begin{align*}
	\bm{e_1}^\top \mathbf{Q} \bm{ f}
	&=
	\frac{\bm{e_1} ^\top \mathbf{G}\bm{e_1}}{\bm{e_1}^\top \mathbf{G}\bm{e_1}\|{\rm Tr}(\mathsf{B})\|}
	=1/\|{\rm Tr}(\mathsf{B})\|\\
	\bm{e_{-1}}^\top \mathbf{Q} \bm{ f}
	&=
	\frac{\bm{e_{-1}} ^\top \mathbf{G}\bm{e_1}}{\bm{e_1}^\top \mathbf{G}\bm{e_1}\|{\rm Tr}(\mathsf{B})\|}\\
	&=-(\mathbf{K}^\top  \mathrm{W}\mathbf{K} + \gamma \mathsf{I})^{-1}\mathbf{K} ^\top \mathrm{W} \mathbf{d}/\|{\rm Tr}(\mathsf{B})\|\\
	&=-\mathbf{U}\mathbf{d}/\|{\rm Tr}(\mathsf{B})\|.
	\end{align*}
	\endgroup
	This completes the proof. 
\end{proof}

\begin{lem} 
	\label{lemb2}
	Denote 
	\begin{align}
	\label{s}
	\mathbf{S} =  \left(\Sigma + \frac{\gamma}{n} \mathsf{I}\right)^{-1}.
	\end{align}
We have
	\begingroup
\allowdisplaybreaks
\begin{align}
&\frac{{\rm d}~\!\bm{g}^\top \mathbf{S} \bm{g}}{{\rm d}~\! \gamma}
=-\frac{\bm{g}^\top \mathbf{S}^{2} \bm{g}}{n} ,
~\frac{{\rm d}~\!\bm{g}^\top \mathbf{S} \bm{g}}{{\rm d}~\! \gamma}
=-\frac{2\bm{g}^\top \mathbf{S}^{3} \bm{g}}{n} \\
&\frac{{\rm d} \tr\big(\mathbf{S} \big) }{{\rm d}~\! \gamma}
=-\frac{\tr\big(\mathbf{S}^{2} \big)}{n} ,
~\frac{{\rm d} \tr\big(\mathbf{S}^{2} \big) }{{\rm d}~\! \gamma}
=-\frac{2\tr\big(\mathbf{S}^{3} \big)}{n}. 
\end{align}
\endgroup
where $\bm{g}$ is any column vector.
	
\end{lem}
The proof is carried out by making use of the matrix differentiation formulas in \citet[Chapter 2]{Petersen2012} and is omitted due to limited space.

\begin{lem} 
	\label{lemb3}
	We have
		\begingroup
	\allowdisplaybreaks
	\begin{align*}
\tr\big( \Sigma^{-1} \big)
&=\Upsilon^{-1}
-\frac{\tr(\mathsf{B})^\top \Upsilon^{-2} \tr(\mathsf{B}) }{ \tr(\mathsf{B})^\top  \Upsilon^{-1}\tr(\mathsf{B})}\\
\bm{\alpha} ^\top \Sigma^{-1}\bm{\alpha} 
&=\bm{\theta}^\top 
\Upsilon^{-1}
\bm{\theta}
-
\frac{\bm{\theta} ^\top \Upsilon^{-1} \tr(\mathsf{B})  \tr(\mathsf{B})^\top \Upsilon^{-1}  \bm{\theta}}
{ \tr(\mathsf{B})^\top  \Upsilon^{-1}\tr(\mathsf{B})}.
\end{align*}
	\endgroup
\end{lem}

\begin{proof}
	By letting $\gamma = 0$ in Lemma \ref{lemb1}, we have
	\begin{align}
	\mathbf{Q} \mathbf{H}
	=
	\left[
	\begin{array}{c}
	0 \\
	(\mathbf{K}^\top  \mathrm{W}\mathbf{K})^{-1}
	\mathbf{K}^\top\mathrm{W}
	\end{array}
	\right]
	\end{align}
	where
	\begin{align*}
	\mathbf{H}
	&=C\mathbf{A}^\top  \mathrm{W} 
	-C \tr(\mathsf{B}) \nonumber
	\frac{ \tr(\mathsf{B})^\top  C \mathbf{A}^\top  \mathrm{W} }{ \tr(\mathsf{B})^\top  C\tr(\mathsf{B})}\\
	C&=\big(\mathbf{A}^\top  \mathrm{W} \mathbf{A}\big)^{-1}.
	\end{align*}
	It follows that
	\begin{align*}
		& \hspace{5mm}
		\mathbf{Q}^\top
		\begin{bmatrix}
	0&0\\
	0&(\mathbf{K}^\top  \mathrm{W}\mathbf{K})^{-1}
	\end{bmatrix}
	\mathbf{Q}
	=
	\mathbf{H}  \mathrm{W}^{-1}  \mathbf{H}^\top
	\\
	&= 
	\left(  \mathsf{I} -  
	\frac{C \tr(\mathsf{B})  \tr(\mathsf{B})^\top   }{ \tr(\mathsf{B})^\top  C\tr(\mathsf{B})}\right)
	C
	\left(  \mathsf{I} -  
	\frac{ \tr(\mathsf{B})  \tr(\mathsf{B})^\top  C }{ \tr(\mathsf{B})^\top  C\tr(\mathsf{B})}\right)\\
	&=
	\left(C-
	\frac{C \tr(\mathsf{B})  \tr(\mathsf{B})^\top C  }{ \tr(\mathsf{B})^\top  C\tr(\mathsf{B})}\right)
	\end{align*}
	which derives that
	\begin{align*}
	\tr\big( \Sigma^{-1} \big)
	&=n\tr\big( (\mathbf{K}^\top  \mathrm{W}\mathbf{K})^{-1} \big)\\
	&=n\tr\big(C \big)
	-\frac{n \tr(\mathsf{B})^\top C^2 \tr(\mathsf{B}) }{ \tr(\mathsf{B})^\top  C\tr(\mathsf{B})}\\
	&=	\tr\big(\Upsilon^{-1}\big)
	-\frac{\tr(\mathsf{B})^\top \Upsilon^{-2} \tr(\mathsf{B}) }{ \tr(\mathsf{B})^\top  \Upsilon^{-1}\tr(\mathsf{B})}\\
	\bm{\alpha} ^\top \Sigma^{-1}\bm{\alpha} 
	&=n 	\bm{\alpha} ^\top(\mathbf{K}^\top  \mathrm{W}\mathbf{K})^{-1} 	\bm{\alpha}\\
	&=
	n
	\bm{\theta}^\top \mathbf{Q}^\top
	\begin{bmatrix}
	0&0\\
	0&(\mathbf{K}^\top  \mathrm{W}\mathbf{K})^{-1}
	\end{bmatrix}
	\mathbf{Q}\bm{\theta}\\
	&=
	n
	\bm{\theta}^\top 
	\left(C-
	\frac{C \tr(\mathsf{B})  \tr(\mathsf{B})^\top C  }{ \tr(\mathsf{B})^\top  C\tr(\mathsf{B})}\right)
	\bm{\theta}\\
	&=\bm{\theta}^\top 
	 \Upsilon^{-1}
	 \bm{\theta}
     -
	\frac{\bm{\theta} ^\top \Upsilon^{-1} \tr(\mathsf{B})  \tr(\mathsf{B})^\top \Upsilon^{-1}  \bm{\theta}}
	{ \tr(\mathsf{B})^\top  \Upsilon^{-1}\tr(\mathsf{B})}
	\end{align*}
This completes the proof.
\end{proof}


\begin{lem}\cite[Theorem 8.2]{Ljung1999}
	\label{ct}
	Assume that
	\begin{enumerate}[1)]
		\item $\mathscr{C}(\gamma)$ is a deterministic function that is continuous in $\gamma \in \Omega$ and  minimized at the point $\gamma^\star$,
		where $\Omega$ is a compact subset of $\mathbb{R}$.
		
		\item A sequence of functions $\{\mathscr{C}_n(\gamma)\}$ converges to $\mathscr{C}(\gamma)$ almost surely and uniformly in $\Omega$ as $n$ goes to $\infty$.
	\end{enumerate}
	Then $\widehat{\gamma}_n = \argmin_{\gamma\in \Omega} \mathscr{C}_n(\gamma)$
	converges to $\gamma^\star$ almost surely, namely,
	\begin{align*}
	|\widehat{\gamma}_n-\gamma^\star|\xra{}0~~\mbox{as}~n\xra{}\infty.
	\end{align*}
\end{lem}


\end{document}